\documentclass[11pt,a4paper]{article}
\usepackage[margin=2.5cm]{geometry}
\usepackage[colorlinks,linkcolor=blue,citecolor=blue,urlcolor=blue]{hyperref}
\usepackage{fancyhdr,lastpage,verbatim,rotating,amsfonts,color}
\usepackage{setspace}
  \usepackage[USenglish]{babel}
  \usepackage{amsmath,amssymb,amsthm}
  
\usepackage{tikz}\usetikzlibrary{patterns}
\usepackage{enumitem}
\usepackage{caption}
\usepackage{subcaption}

\usepackage[numbers]{natbib}
\usepackage{breakcites}

\newcommand{\R}{\mathbb{R}}
\newcommand{\G}{\mathcal{G}}

\newcommand{\HA}{\mathcal{H}}

\newcommand{\E}[2]{\mathbb{E}_{#1}\left[#2\right]}

\newcommand{\Exp}[1]{\mathbb{E}\left[#1\right]}
\newcommand{\bx}{\bar{x}}

\newcommand{\ut}{u}

\newcommand{\argmin}{\operatornamewithlimits{argmin}}

    \usepackage{hyperref}
\newtheorem{assumption}{Assumption}
\newtheorem{theorem}{Theorem}[section]
\newtheorem{lemma}{Lemma}[section]
\newtheorem{corollary}{Corollary}[section]

\newtheorem{proposition}{Proposition}[section]
\newtheorem{definition}{Definition}[section]

\newtheorem{remark}{Remark}[section]

\title{Robust Quantitative Comparative Statics\\ for a Multimarket Paradox}

\author{
Tobias Harks\thanks{
University of Augsburg, Institute of Mathematics, 86135 Augsburg, Germany
Email: \texttt{tobias.harks@math.uni-augsburg.de}}
\and
Philipp von Falkenhausen\thanks{Technical University Berlin, Institute of Mathematics, 10623 Berlin, Germany Email: \texttt{falkenhausen@math.tu-berlin.de}}
}

\begin{document}

\maketitle
\begin{abstract}%
We introduce a quantitative approach to comparative statics that allows to bound the maximum effect of an exogenous parameter change on a system's equilibrium.
The motivation for this approach is a well known paradox in multimarket Cournot competition,
where a positive price shock on a monopoly market may actually reduce the monopolist's
profit.
We use our approach to quantify for the first time the worst case profit
reduction for multimarket oligopolies
exposed to arbitrary positive price shocks.
For markets with affine price functions and firms with convex cost technologies,
we show that the relative profit loss
of \emph{any} firm is at most $25\%$ no matter
how many firms compete in the oligopoly. We further investigate the impact of positive price shocks
on total profit of all firms as well as on
social welfare. We find tight bounds also for these measures
showing that total profit and social welfare decreases by at
most $25\%$ and $16.6\%$, respectively. 
Finally, we show that in our model, mixed, correlated and coarse
correlated equilibria are essentially unique, thus, all 
our bounds apply to these game solutions as well. 
\end{abstract}



\section{Introduction}

Many economic settings can be modeled as strategic games $(N,X,u)$, where $N$ denotes a set of players, $X$ is the set of possible states of the games (all possible decisions by all players) and $u:X\rightarrow \R^N$ denotes the players' payoff functions, i.e. in state $x\in X$ player $i$ receives payoff $u_i(x)$. 
Given an instance of such a model, there are parameters that determine the precise values for $N$, $X$, and $u$. 
Comparative statics is concerned with the effect of changes in these parameters on the equilibria of the game. 
Examples of such changing parameters are exposure to international trade~(\citet{krugman1980,melitz2003}) or a forced reduction of produced quantity~(\citet{gaudet1991}) (both affect the state space) or changes of the payoff functions via a demand shift~(\citet{quirmbach1988}), a cost shift~(\citet{fevrier2004}), or the introduction of export taxes/subsidies~(\citet{brander1985,eaton1986}).
A classical approach to comparative statics analysis uses the implicit function theorem\footnote{The application of
 the implicit function theorem for instance in oligopoly models
 requires some regularity assumptions
 such as convexity and smoothness of cost and inverse demand functions, see the discussion in
\citet{MR94}.}(cf.~\citet{Bulow85}). 
This type of analysis, however, is restricted to small (local) parameter changes, and, thus,  does not allow for meaningful conclusions if parameter
changes happens to be discrete. Moreover, it requires knowledge of the precise instantiation of the 
model (in contrast
to knowing only some functional forms of the models, e.g., utility functions are concave).
A more recent approach to comparative statics is based on the
powerful machinery of lattice theory applicable to supermodular games, see 
\citet{Amir96}, 
\citet{Edlin98}, 
\citet{Milgrom94}, 
\citet{Milgrom90,MR94}, 
\citet{Shannon95} and 
\citet{Topkis79, Topkis98}, also  
\citet{Athey02}, 
\citet{Quah07}, and 
\citet{Vives90,Vives05}.
Common to all these works is that they consider the 
monotonicity of effects, i.e., whether a certain objective function increases or 
decreases with a parameter, but they do not quantify the effect.

Our approach is to capture the \textit{maximum possible} effect that the change of a parameter -- shift of the inverse demand function in our case -- can have. 
This worst-case approach exhibits both  
\begin{enumerate}[noitemsep]
    \item significance: are changes in a given parameter worth considering?
    \item robustness: how sensitive is the game to changes of a parameter?
\end{enumerate}
Significance is a crucial motivation of both the analysis of an effect and discussion of whether it can be put to use (\`a la `should a new tax be introduced?'). Robustness on the other hand is important when there is uncertainty about the values of parameters and when parameters change over time. 

\paragraph{A Paradox in Multimarket Cournot Oligopolies.}
We apply our approach to the multimarket oligopoly model introduced by 
\citet*{Bulow85}. They
investigated how ``changes in one market have
ramifications in a second market'' and discovered that a positive price
shock in a firm's monopoly market can have a negative
effect on the firm's profit by influencing competitors' strategies in a different market. This
counterintuitive phenomenon led them to the classification of markets in terms of strategic substitutes and strategic complements.\footnote{See Section~VII in~\citet{Bulow85}.}\footnote{In a market
with \emph{strategic substitutes}, more aggressive play by a firm
leads to less aggressive play of the competitors on that market;
with \emph{strategic complements}, more aggressive play
results in more aggressive play of the competitors.} 
Our paper is about rigorously quantifying profit effects
induced by price shocks in multimarket Cournot oligopolies.

Let us recall the example of two markets $\{1,2\}$ and two firms $\{a,b\}$ given in~\citet{Bulow85}.
Firm $a$ is a monopolist on market $1$ and competes
with firm $b$ on market $2$. Demand is infinitely elastic on market
$1$ for the constant price $p_1=50$. On market $2$, there is an affine price function
given by $p_2(q_{a,2}+q_{b,2})=200-(q_{a,2}+q_{b,2})$, where
$q_{a,2},q_{b,2}$ denote the quantities offered by the respective
firms on market $2$. Production costs are symmetric
and given by $c(q)=1/2q^2$, where $q$ is the total
quantity produced by a firm.\footnote{
\citet{Bulow85} assume
additional fixed cost $F>0$ to prevent firms
from setting up multiple plants. Fixed costs, however, 
do not change the resulting equilibria assuming that
the access to markets is exogenously determined.}
In the Cournot equilibrium, we obtain $q_{a,1}=0$ and $q_{a,2}=q_{b,2}=50$
and each firm earns profits $3750$.

Suppose now that market $1$ experiences a positive price shock
raising its constant price by five units to $55$. The Cournot equilibrium
changes to $q_{a,1}=8$ and $q_{a,2}=47$, $q_{b,2}=51$.
Under this new equilibrium, firm $b$ increases its
profit to $3901.5$ while
firm $1$ obtains after the price shock a profit of $3721.5$.
As noted by 
\citeauthor{Bulow85}, the "positive" price
shock to market $1$ has hurt $a$ even though $a$ sells
more units than before and receives on its monopoly market even at a higher
price. The actual profit reduction for firm $a$ amounts to $0.76\%$. 
A natural question to ask is: how much can a firm
lose from a positive price shock on its monopoly
market?

\paragraph{Robust Quantitative Comparative Statics.}
To address the above question, we propose a robust quantitative approach 
to comparative statics. We aim at quantifying the \textit{maximum possible} effect that the change of a parameter
can have 
on a given set of games.

Following the framework of 
\citet{MR94},
we denote by $\G$ a \emph{class} of models or games describing the 
information set of what the modeler knows about the economic environment. 
Suppose there is an objective function $f:\G\rightarrow \R$ (e.g., welfare of the unique equilibrium outcome) to evaluate these games and denote for any $G\in\G$ by $\Delta_G$ all parameter changes that are to be considered. We express the effect of a parameter change $\delta\in\Delta_G$ on a game $G\in\G$ (assuming for the moment $f(G)\neq 0$) as the value
\[\gamma^f(G,\delta)=\frac{f(G(\delta))}{f(G)},\]
where $G(\delta) \in \G$ denotes the changed game. 
For $f(G)=0$, we set $\gamma^f(G,\delta)=1$, if $f(G(\delta))\geq 0$, 
and  $\gamma^f(G,\delta)=-\infty$, if $f(G(\delta))<0$.
The maximum possible effect across all games in $\G$ and their respective parameter changes $\Delta_\G$
is then defined as
\[\gamma^f(\G,\Delta_\G)=\inf_{G\in\G}\inf_{\delta\in\Delta_G}\gamma^f(G,\delta).\]
Note that if we are interested in the opposite direction of $f$ we can replace the infima with suprema in the above definition.

Quantifying $\gamma^f(\G,\Delta_\G)$ can expose the significance and robustness of the effect of parameter changes, contrasting comparative statics analysis that only reveals the monotonicity, as outlined above. 
Besides a quantitative instead of a qualitative result, 
there are further key differences to prevailing comparative statics approaches:

In practice, the economic model including the endogenous variables and parameters may not be known precisely but only approximately by knowing some functional property  of the relations of endogenous variables or the space of parameters.\footnote{
\citet{MR94}
refer to the set of possible instantiations of the model as a \emph{class} of models, or the \emph{context} ``representing what the modeler knows about the economic environment''.}
The nature of a worst-case analysis such as the proposed robust quantitative comparative statics is that it provides results across a class of games $\G$ described for example by functional properties (e.g. any convex cost technologies in our case) and a set $\Delta_\G$ of parameter changes.
Both $\G$ and $\Delta_\G$ need to be chosen reasonably to yield a meaningful result (any combination of positive price shocks in our case). 
 This is contrasted by the assumption of \emph{perfect knowledge} of all elements of the precise instance $G$
 including the precise parameter change $\delta$
common to most existing approaches in comparative statics, e.g. those based on the implicit function theorem.


\paragraph{Our Results.}
We conduct robust quantitative comparative statics analysis 
for multimarket oligopolies with affine price functions
and convex cost technologies of firms, a class of games that contains the above example. We consider positive price shocks as parameter changes 
and
three different objective functions: the individual profit of a firm, total profit measured by summing up the firm's profits, and social welfare defined by integrating the price functions and subtracting the firm's costs. 
We find that both profit and total profit can be reduced at most $25\%$ by a positive price shock. For profit we give the bound as a function  of the number of firms, showing e.g. that in the two firm case the profit reduction is a most $6.25\%$. Social welfare on the other hand can only be reduced by up to $16.7\%$ by a positive price shock. These results immediately extend to negative price shocks, as any negative price shock can be seen as taking back a positive price shock. The profit and welfare gain from a negative price shock is no more than $33.4\%$, and the social welfare can increase by up to $20\%$.
Our results give the first robust quantitative comparative statics result of an important paradoxical phenomenon
previously only qualitatively analyzed. We further show that for the considered model, mixed, correlated and coarse
correlated equilibria are essentially unique, thus, all 
these worst-case bounds apply to these solution concepts as well. 

Each of our upper bounds is complemented by an instance that actually attains the worst case ratio. We demonstrate that worst case instances are quite generic in the sense that only
a set of rather non-restrictive properties needs to be satisfied. Thus, 
a common criticism of worst case analysis that isolated special cases lead to unrealistic objective function values does not apply to our setting.
If linearity of price functions (or convexity of cost functions) are relaxed, our bounds do not hold anymore. For concave price functions (assuming still convex cost functions), we give an example
in which the profit of a firm is eliminated after a positive price shock. For non-convex
cost functions we sketch a similar example.

\paragraph{Example Application.}
Profit gains from negative price shocks
can occur in international trade, as noted by 
\citet[Sec. VI (C)]{Bulow85}.
Consider two markets located in separate countries with 
convex cost technologies, one of which is a monopoly market for firm $a$. 
A tax change in the country of the monopolist can be considered a price shock.
A government may decide to increase domestic taxes
in order to increase firm $a$'s profitability in the foreign market. 
Our results imply that this positive effect can indeed be significant as it may increase the profitability
by up to $33\%$ of current profits.


\section{Related Literature}\label{sec:related}
Comparative statics has long been used to understand how a system is affected by changes of underlying parameters, see 
\citet{dixit1986} for an overview.

Since 
\citet{Bulow85} introduced the concepts of strategic substitutes and strategic complements, a considerable amount of research has used comparative statics to investigate situations where strategic substitutes or complements occur. Notably, 
\citet{brander1985} found that an export subsidy can increase welfare in Cournot competition with strategic substitutes (although they do not name the concept). 
\citet{eaton1986} extended this model to a two-stage game where first governments set policies and then firms engage in competition. 
\citet{gaudet1991} looked at how a forced marginal change of production quantity for a subset of firms impacts profit in a situation with strategic substitutes. 
\citet{fevrier2004} gave a decomposition of price shocks in Cournot oligopoly into an average effect common to all firms and a heterogeneity effect that is firm-specific. 
\citet{acemoglu2013} present a framework for comparative statics results for a superclass of Cournot oligopoly called aggregative games, see also 
\citet{corchon1994} and \citet{Kukushkin94}.

Worst case perspectives have a long tradition in the analysis of performance guarantees of algorithms 
(cf.~\citet{williamson2011}). With the introduction of concepts like the ``price of anarchy'' (cf. \citet{Koutsoupias99}) they have found their way into economics and game theory (a good
overview on the intersection of computer science and economics is given in the textbook~\citet{Nisan:2007}). 
By now, there is a large body of literature quantifying the worst case efficiency loss of equilibria
(cf.~\citet{Roughgarden02,Johari04,CorSS04})
and, in particular,
 in models related to Cournot competition (cf.~\citet{AndersonR03,FarahatP09,FarahatP11,Guo05,Johari05, KlubergP12, TsitsiklisX13,Tsitsiklis14}).
All these works follow the ``price of anarchy'' methodology, where
some performance measure of an equilibrium outcome (e.g., total profit and/or social welfare) 
is compared with a
corresponding optimal solution maximizing the respective performance measure.
There is a conceptual difference of our approach compared to these previous worst case approaches: while in approximation algorithms and price of anarchy one compares an optimum to a solution that is returned by an algorithm or that is an equilibrium, we maintain an equilibrium concept and compare the equilibrium of an instance with that of a perturbed instance. Of course, in our setting,
a bound on the price of anarchy (with respect to a specific welfare function)
translates to a bound on the worst-case effect of a parameter change
as long as the optimal welfare increases with the parameter change (which is true
in our case). However, the price of anarchy in cournot oligopolies, e.g., for total welfare (involving affine price functions and convex costs) is
known to be $4/(3+n)$, where $n$ is the number of firms 
(cf.~\citet{Harks_Miller2011,Moulin05}).
Also for social welfare the known bound of $2/3$ (cf.~\citet{Johari05})
does not give tight bounds for our setting.
In the context of multimarket Cournot oligopoly, to the best of our knowledge, we derive for the first time a worst case quantification of a parameter change previously only qualitatively described by comparative statics.  

There are works in theoretical computer science and mathematics that  quantify the effects of a parameter change, albeit for different settings. Prominent examples include the analysis of the famous Braess paradox, where a limit on the decrease in a network's performance when additional links are added is shown 
(cf.~\citet{Korilis99,LinRTW11,Roug06,Valiant10}).

\section{Multimarket Cournot Oligopoly Model}\label{sec:model}
In this section we introduce the specific model for which we investigate robust quantitative comparative statics. As outlined above, the three ingredients for such analysis are a set of games $\G$, for each such game $G\in\G$ a set of feasible parameter changes $\Delta_G$ and an objective function $f:\G\rightarrow\R$ to evaluate the games.
\paragraph{Class of games $\G$.}
We consider multimarket oligopolies defined as follows:
In a game $G\in\G$ a set 
$N$ of $n$ firms competes in a set of markets $M$. 
Each firm~$i\in N$ has access to some subset $M_i\subseteq M$ of these markets. A strategy of a firm~$i$ is to choose  nonnegative production quantities $q_{i,m}$ for all markets $m\in M_i$ that it serves. We set $q_{i,m}=0$ for any market not served by firm $i$ and denote the total quantity of firm $i$ by $q_i = \sum_{m\in M}{q_{i,m}}$. Correspondingly, the total quantity on any market $m$ is denoted by $q_m=\sum_{i\in N}{q_{i,m}}$.

\begin{assumption}\label{ass1}
 The price of a market is an affine function of the total quantity produced, i.e., if on market $m\in M$ a quantity $q_m$ is produced, the price is $p_m(q_m)=s_m - r_m q_m,$ where $s_m>0$ is an initial positive price and the coefficient $r_m>0$ describes how the price decreases as the total quantity is increased.
\end{assumption}
We sometimes denote the price function on a market simply by $p_m(q)$. 
\begin{assumption}\label{ass2}
Firm $i$'s cost for producing the quantity $q_i$ 
is given by the function $c_i(q_i)$ which we assume to be nondecreasing, convex and differentiable in $q_i$ with $c_i(0)=0$. 
\end{assumption}

Given production quantities for all firms, the utility of firm $i$ is defined as \[ u_i(q) := \sum_{m\in M}{ p_m(q_m) q_{i,m}} - c_i(q_{i}).\]  
In a Cournot equilibrium, firms choose their quantities so as to maximize their utility. Particularly, in a Cournot equilibrium $q$ no firm can increase its utility by unilaterally deviating to a different strategy, i.e., $\ut_i(q)\geq \ut_i(q'_i,q_{-i})$ for all strategies $q'_i$ available to firm $i$. 
As the games introduced above
 involve convex and compact strategy spaces together with quasi-concave
 utilities, standard fixed-point arguments of type Kakutani (cf.~\cite{Debreu52,Fan52,Glicksberg52,Kakutani41}) 
 imply the existence of an equilibrium. Note that compactness of the strategy space
 can be guaranteed, because for every firm and every market there exists a maximum quantity
 at which the market price goes down to zero eliminating the firm's profit. Thus, one can effectively bound all quantities by a finite value. As shown in Section~\ref{sec:profit}, the assumptions
 on the price and cost functions further imply that there is a unique equilibrium.

We denote the marginal revenue of firm $i$ on market $m$ by 
\[\pi_{i,m}(q_{i,m},q_{-i,m})=\frac{\partial\left(p_m(q_{i,m}+q_{-i,m})q_{i,m}\right)}{\partial q_{i,m}} 
= p_m(q_m)-r_mq_{i,m},
\]
where $q_{-i,m}$ is the quantity sold by firms $j\neq i$. The marginal cost is $c_i'(q_i)$, and we will often write $\pi_{i,m}(q)$ and $c_i'(q)$. In an equilibrium $x$, the marginal revenue of any served market $m$ equals the marginal cost: $\pi_{i,m}(x)=c'_i(x)$ for all $i\in N$ with $x_{i,m}>0$.
\paragraph{Parameter changes $\Delta_G$.}
The parameter changes analyzed are positive shocks to the price functions: 
on every market $m$ of a game $G\in\G$ the price increases by some amount $\delta_m\geq 0$ such that $p_m^{\delta}(q_m)=p_m(q_m) + \delta_m$. 
We denote the set of feasible parameter changes as 
 $\Delta_G:=\R^{|M|}_{\geq 0}$, where $|M|$ is the number of markets in $G$.

\paragraph{Objective functions $f$.}
The games in $\G$ have unique Cournot equilibria (as shown in Section \ref{sec:profit} and Section~\ref{sec:unique}), and our three objective functions are equilibrium firm profit, equilibrium total profit and equilibrium social welfare.
Given a game $G$ and a price shock $\delta$, we denote the unique equilibrium of $G$ by $x$ and the unique equilibrium of $G(\delta)$ by $y$.
\begin{enumerate}[noitemsep]
    \item firm profit: the profit of an individual firm is $u_i(q)= \sum_{m\in M}{ p_m(q_m) q_{i,m}} - c_i(q_{i})$ and we minimize across the firms of a game to obtain an overall measure: 
    \[\gamma^u(G,\delta):=\min_{i\in N}{\frac{\ut^\delta_i(y)}{\ut_i(x)}}.\]
    \item total profit: we consider utilitarian welfare, i.e. the sum of the profits of all firms $U(q):=\sum_{i\in N}{u_i(q)}$, such that
    \[\gamma^U(G,\delta):=\frac{U^\delta(y)}{U(x)}.\]
    \item social welfare: this measure assumes that the price function of a market expresses the marginal value that the buyers in the market have from additional quantity of the good. We denote $S(q):=\sum_{m\in M}\int_{0}^{q_m}p_m(z)dz-\sum_{i\in N}c_i(q_i)$, such that
    \[ \gamma^S(G,\delta):=\frac{S^\delta(y)}{S(x)}.\]
\end{enumerate}
While the first measure has been analyzed by 
\citet{Bulow85},
the second two measures have been used among others by 
\citet{AndersonR03}, 
\citet{Ushio85} and 
\citet{TsitsiklisX13}. For each measure, we strive to provide a infimum bound across all multimarket Cournot games in $\G$ and price shocks $\Delta_\G$ combined with concrete games that converge to this bound. 
The following example instance exhibits the basic intuition for the quantitative analysis in Sections~\ref{sec:profit} and~\ref{sec:aggregate}.

\paragraph{Example instance. }
 Consider a game with two markets. Market $1$ is served only by a monopolist (denoted as firm $a$), while all firms compete in market $2$.
Given a positive price shock
 on market $1$, firm $a$ will reduce its quantity on market $2$
 in favor of selling more on its more profitable monopoly market, see Figure~\ref{fig:ex1}.
In effect, the competitors' marginal revenue strictly increases and leads
to a new equilibrium in which these competitors increase their quantities, see Figure~\ref{fig:ex2}.
Markets in which
a less aggressive play of one firm leads to a more aggressive
play of competitors are called markets with strategic substitutes, see~\citet{Bulow85}.
The more aggressive competition experienced by firm $a$ in market~$2$ 
reduces its profit below the original level, or said differently, 
after the price shock a part of firm $a$'s quantity on market $2$ has
been substituted by the competitors, see Figure~\ref{fig:ex3}.

\paragraph{Negative Price Shocks.}\label{sec:negative}
While we restrict our analysis to positive price shocks, the results immediately extend to negative price shocks. 
If for example a subsidy (i.e. a positive price shock) causes the equilibrium 
to shift such that the firms' profits decrease, then taking back the subsidy (i.e. a negative price shock) will restore the old equilibrium and thus increase the firm's profits. In this sense the two effects are dual: any negative price shock can be seen as taking back a positive price shock, and the profit gain from the negative price shock is equal to the profit loss from positive price shock. This is true for any objective function $f$ and, denoting negative price shocks to a game by ${\Delta^-_\G}$ we have $\gamma^f(\G,{\Delta^-_\G})=(\gamma^f(\G,\Delta_\G))^{-1}$.

\begin{figure}\label{fig:ex}
\caption{Example instance with two markets}
\centering
\begin{subfigure}[b]{1\textwidth}
\centering 
\begin{tikzpicture}[scale=1]
    \draw [<->] (0,3.5) node (yaxis) [above] {}
        |- (5,0) coordinate (xaxis);
    \draw (xaxis) node[below] {$q_{a}$};
    \coordinate (pi_1) at (0,3);
    \coordinate (pi_2) at (3,0);
    \coordinate (p1) at (0,0.75);
    \coordinate (p2) at (4,0.75);
    \coordinate (i) at (intersection of pi_1--pi_2 and p1--p2);
    \draw[thick] (pi_1) -- node[sloped,above,near start] {$\pi_{a,2}(q_{a,2},x_{-a})$} (pi_2) ; 
    \draw[ultra thick] (pi_1) -- (i) ; 
    \draw[ultra thick] (i) -- (2.7,0.75) coordinate (i2) ; 
    \draw[thick] (p1) -- (p2) node[right]{$p_1$}; 
    \draw[loosely dotted] (i)  -- (i |- xaxis) node[below] {$x_{a,2}$};
    \draw[loosely dotted] (i2)  -- (2.7,-.3) node[below] {$x_{a,2}+x_{a,1}$};
    \coordinate (c_1) at (0,0.2);
    \coordinate (c_1a) at (2.5,0.6);
    \coordinate (c_2) at (4.3,2.5);
    \draw [thick] plot [smooth] coordinates {(c_1) (c_1a) (c_2)} node[above, align=left]{$c_a'(q_{a})$}; 
\end{tikzpicture}
\begin{tikzpicture}[scale=1]
    \draw [<->] (0,3.5) node (yaxis) [above] {}
        |- (5,0) coordinate (xaxis);
    \draw (xaxis) node[below] {$q_{i}$};
    \coordinate (pi_1) at (0,1.5);
    \coordinate (pi_2) at (1.5,0);
    \draw[thick] (pi_1) -- node[sloped,above] {$\pi_{i,2}(q_{i},x_{-i})$\ \ \ \ \ \ } (pi_2) ; 
    \coordinate (i) at (1.2,0.3);
    \draw[dotted] (i)  -- (i |- xaxis) node[below] {$x_{i,2}$};
    \draw (2.7,-.3) node[below] {\phantom{$x_{a,2}+x_{a,1}$}}; 
    \coordinate (c_1) at (0,0.2);
    \coordinate (c_1a) at (1.7,0.5);
    \coordinate (c_2) at (3.2,2.5);
    \draw [thick] plot [smooth] coordinates {(c_1) (c_1a) (c_2)} node[above, align=left]{$c_i'(q_{i})$}; 
\end{tikzpicture}
\caption{\textit{Initial equilibrium $x$:} The price on market $1$ is constant at $p_1$ (thus also firm $a$'s marginal revenue). On market $2$ the price is decreasing, such that firm $a$ sees marginal revenue $\pi_{a,2}(\cdot,x_{-a})$, given its competitors' equilibrium quantities $x_{-a}$. Firm $a$ produces $x_{a,2}$ on market $2$ such that the marginal revenue there is equal to the marginal revenue $p_1$ on market $1$. Its production on market $1$ is such that the marginal cost of the aggregate quantity $x_{a,1}+x_{a,2}$ is equal to $p_1$. The competitors $i\neq a$ have marginal revenue $\pi_{i,2}(\cdot,x_{-i})$ on market $2$. They choose their equilibrium quantity $x_{i,2}$ at the intersection of marginal cost and marginal revenue.}
 \label{fig:ex1}
\end{subfigure} 

\begin{subfigure}[b]{1\textwidth}
\centering 
\begin{tikzpicture}[scale=1]
    \draw [<->] (0,3.5) node (yaxis) [above] {}
        |- (5,0) coordinate (xaxis);
    \draw (xaxis) node[below] {$q_{a}$};
    \coordinate (pi_1) at (0,3);
    \coordinate (pi_2) at (3,0);
    \coordinate (p1) at (0,0.75);
    \coordinate (p2) at (4,0.75);
    \coordinate (p1delta) at (0,1.3);
    \coordinate (p2delta) at (4,1.3);
    \coordinate (i) at (intersection of pi_1--pi_2 and p1--p2);
    \coordinate (ix) at (intersection of pi_1--pi_2 and p1delta--p2delta);
    \coordinate (i2) at  (2.7,0.75); 
    \draw[thin,lightgray] (p1) -- (p2) node[right]{$p_1$}; 
    \draw[loosely dotted,lightgray] (i)  -- (i |- xaxis);
    \draw[loosely dotted,lightgray] (i2)  -- (i2 |- xaxis);
    \draw[thick] (pi_1) -- node[sloped,above,near start] {$\pi_{a,2}(q_{a,2},x_{-a})$} (pi_2) ; 
    \draw[thick] (p1delta) -- (p2delta) node[right]{$p_1+\delta$}; 
    \coordinate (c_1) at (0,0.2);
    \coordinate (c_1a) at (2.5,0.6);
    \coordinate (c_2) at (4.3,2.5);
    \draw [thick] plot [smooth] coordinates {(c_1) (c_1a) (c_2)} node[above, align=left]{$c_a'(q_{a})$}; 
    \draw[ultra thick] (pi_1) -- (ix) ; 
    \draw[ultra thick] (ix) -- (3.3,1.3) coordinate (i2x); 
    \draw[loosely dotted] (ix)  -- (ix |- xaxis);
    \draw[loosely dotted] (i2x)  -- (i2x |- xaxis);
    \draw (ix |- xaxis) node[below right] {$\leftarrow$};
    \draw (i2x |- xaxis) node[below left] {$\rightarrow$};
\end{tikzpicture}
\begin{tikzpicture}[scale=1]
    \draw [<->] (0,3.5) node (yaxis) [above] {}
        |- (5,0) coordinate (xaxis);
    \draw (xaxis) node[below] {$q_{i}$};
    \coordinate (pi_1) at (0,1.5);
    \coordinate (pi_2) at (1.5,0);
    \coordinate (pishift_1) at (0,3);
    \coordinate (pishift_2) at (3,0);
    \draw[thick] (pi_1) -- (pi_2) ; 
    \draw[thick] (pishift_1) -- node[sloped,above,near start] {$\pi_{i,2}(q_{i},y_{-i})$} (pishift_2); 
    \draw[->] (0.2,1.5) -- (0.7,2);
    \draw[->] (0.7,1) -- (1.2,1.5);
    \draw[->] (1.2,0.5) -- (1.7,1);
    \coordinate (i) at (1.2,0.3);
    \coordinate (ix) at (2.1,.95);
    \draw[dotted] (i)  -- (i |- xaxis) node[below] {$x_{i,2}$};
    \draw[loosely dotted] (ix)  -- (ix |- xaxis) node[below] {$y_{i,2}$};
    \coordinate (c_1) at (0,0.2);
    \coordinate (c_1a) at (1.7,0.5);
    \coordinate (c_2) at (3.2,2.5);
    \draw [thick] plot [smooth] coordinates {(c_1) (c_1a) (c_2)} node[above, align=left]{$c_i'(q_{i})$}; 
\end{tikzpicture}
\caption{\textit{Price shock triggers strategic substitution:} The shock leads firm $a$ to shift production from market $2$ to market $1$. This increases the marginal revenue of competitors on market $2$, causing increased production $y_{i,2}>x_{i,2}$.}
 \label{fig:ex2}
\end{subfigure} 

\begin{subfigure}[b]{1\textwidth}
\centering 
\begin{tikzpicture}[scale=1]
    \draw [<->] (0,3.5) node (yaxis) [above] {}
        |- (5,0) coordinate (xaxis);
    \draw (xaxis) node[below] {$q_{a}$};
    \coordinate (pi_1) at (0,3);
    \coordinate (pi_2) at (3,0);
    \coordinate (pishift_1) at (0,1.8);
    \coordinate (pishift_2) at (1.8,0);
    \draw[thick] (pishift_1) -- node[sloped,below] {$\pi_{a,2}(q_{a,2},y_{-a})$} (pishift_2); 
    \coordinate (p1delta) at (0,1.3);
    \coordinate (p2delta) at (4,1.3);
    \coordinate (i) at (intersection of pi_1--pi_2 and p1--p2);
    \coordinate (ix) at (intersection of pi_1--pi_2 and p1delta--p2delta);
    \coordinate (iy) at (intersection of pishift_1--pishift_2 and p1delta--p2delta);
    \coordinate (i2) at  (2.7,0.75); 
    \draw[thick] (pi_1) -- (pi_2) ; 
    \draw[thick] (p1delta) -- (p2delta) node[right]{$p_1+\delta$}; 
    \coordinate (c_1) at (0,0.2);
    \coordinate (c_1a) at (2.5,0.6);
    \coordinate (c_2) at (4.3,2.5);
    \draw [thick] plot [smooth] coordinates {(c_1) (c_1a) (c_2)} node[above, align=left]{$c_a'(q_{a})$}; 
    \draw[fill=lightgray, fill opacity=0.5] (pi_1) -- (ix) -- (iy) -- (pishift_1) -- cycle;
    \draw[ultra thick] (pishift_1) -- (iy) ; 
    \draw[ultra thick] (iy) -- (3.3,1.3) coordinate (i2x); 
    \draw[<-] (0.35,1.65) -- (0.7,2);
    \draw[<-] (1.35,0.65) -- (1.7,1);
    \draw[loosely dotted] (i2x)  -- (i2x |- xaxis) node[below] {$y_{a,2}+y_{a,1}$};
    \draw[loosely dotted] (iy)  -- (iy |- xaxis) node[below] {$y_{a,2}$};
\end{tikzpicture}
\begin{tikzpicture}[scale=1]
    \draw [<->] (0,3.5) node (yaxis) [above] {}
        |- (5,0) coordinate (xaxis);
    \draw (xaxis) node[below] {$q_{i}$};
    \coordinate (pishift_1) at (0,3);
    \coordinate (pishift_2) at (3,0);
    \draw[thick] (pishift_1) -- node[sloped,above,near start] {$\pi_{i,2}(q_{i},y_{-i})$} (pishift_2); 
    \coordinate (ix) at (2.1,.95);
    \draw[loosely dotted] (ix)  -- (ix |- xaxis) node[below] {$y_{i,2}$};
    \coordinate (c_1) at (0,0.2);
    \coordinate (c_1a) at (1.7,0.5);
    \coordinate (c_2) at (3.2,2.5);
    \draw [thick] plot [smooth] coordinates {(c_1) (c_1a) (c_2)} node[above, align=left]{$c_i'(q_{i})$}; 
\end{tikzpicture}
\caption{\textit{Effect on firm $a$'s profit:} This derogates firm $a$'s marginal revenue on market $2$, causing it to further withdraw from the market. The shaded area indicates the profit loss from the strategic substitution, which is partially compensated by  profit gain from the increased price on market $1$.}
 \label{fig:ex3}
\end{subfigure}
\end{figure}
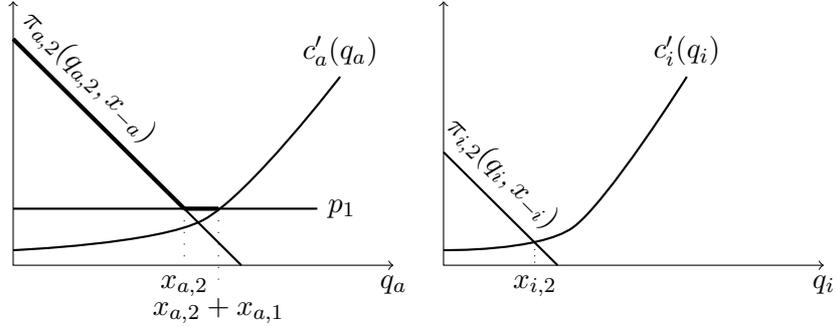 

\section{Maximum Profit Loss of An Individual Firm}\label{sec:profit}
We use robust quantitative comparative statics to investigate the worst case profit loss of a firm from a positive price shock as expressed by $\gamma^u(\G,\Delta_\G)$.
\begin{theorem}\label{thm:main}
Given a game $G$ with $n$ firms, no firm loses more than a $\frac{(n-1)^2}{4n^2}$ fraction of its profit from a price shock  $(\delta_m)_{m\in M}$ with $\delta_m\geq 0$ for all $m\in M$,
\[ \gamma^u(G,\delta)\geq 1-\frac{(n-1)^2}{4n^2}\geq \frac{3}{4}.
\]
\end{theorem}
This is our main result. It shows that no firm loses more than $25\%$
of current profits. This bound is \emph{robust} in the sense
that it holds for an entire class of games and parameters, that is, in order to arrive at this bound the modeler only needs to know that price shocks are nonnegative, inverse demand functions
are affine, and cost technologies are convex.

To prove the statement, we first establish uniqueness of equilibria and that a price shock $\delta\geq 0$ causes the price on every market to increase, and that in this favorable setting every firm increases its total quantity. 
Using these insights into the effect of the price shock, for any given firm $i$ we can
identify a market where this firm suffers the relatively strongest loss and use this to bound $\frac{u_i^\delta(y)}{u_i(x)}$. This proves the theorem, as $\gamma^u(G,\delta)$ is the minimum of this fraction across all firms.
The following lemma follows easily by the fact that our games are exact (concave) potential
games over a convex and compact strategy space, see~\citet{Monderer96}. We give here an alternative combinatorial proof
and some of the ideas will be used later on. 
\begin{lemma}[Uniqueness of equilibria]\label{lem:multiunique}
Let $x$ and $y$ be equilibria of some game $G$. Then $x=y$.
\end{lemma}
\proof
As a first step of the proof, we show that on every market $m$, $x_m = y_m$. Let $M^+:=\{m\in M:x_m<y_m\}$ and assume for a contradiction that $M^+\neq \emptyset$. Then, there is a firm $i$ with $\sum_{m\in M^+}{(y_{i,m}-x_{i,m})}>0$ and a market $m\in M^+$ with $y_{i,m}>x_{i,m}\geq 0$. It follows from the equilibrium definition that in $y$ firm $i$'s marginal cost and marginal revenue on $m$ are equal, i.e.,
\begin{equation}\label{c_iyleqc_ix}
  c_i'(y_i)=p_m(y_m)-r_my_{i,m}<p_m(x_m)-r_mx_{i,m}\leq c_i'(x_i),
\end{equation}
where we used that $p_m(y_m)<p_m(x_m)$ because $m\in M^+$. From $c_i'(y_i)<c_i'(x_i)$ it follows that 
\begin{equation}\label{y_ileqx_i}
  y_i<x_i.
\end{equation}
Also, for all markets $m'\in M$ where $y_{i,m'}<x_{i,m'}$, it follows again from the equilibrium definition that 
\[
  p_{m'}(y_{m'})-r_{m'}y_{i,m'}\leq c_i'(y_i)<c_i'(x_i)\leq p_{m'}(x_{m'})-r_{m'}x_{i,m'},
\]
and hence, $p_{m'}(y_{m'})<p_{m'}(x_{m'})$, i.e. $m'\in M^+$. Then, we find a contradiction as
\[
  0 > y_i - x_i = \sum_{m\in M}{(y_{i,m}-x_{i,m})}>\sum_{m\in M^+}{(y_{i,m}-x_{i,m})}>0.
\]
Here, we can limit the summation from all $m\in M$ to $m\in M^+$ because we found that $m'\in M^+$ for all markets with $y_{i,m'}<x_{i,m'}$.

As the next step of the proof, we use $x_m=y_m$ for all $m\in M$ to show $x_{i,m}=y_{i,m}$ for all firms $i$. For a contradiction, assume there are $i\in N$ and $m\in M$ such that $x_{i,m}<y_{i,m}$. Then, we can again apply~\eqref{c_iyleqc_ix} to obtain $y_i<x_i$ and there must be some market $m'\in M$ with $y_{i,m'}<x_{i,m'}$, which leads with the same reasoning to $x_i<y_i$, a contradiction. 
Altogether, $x=y$.
\endproof

We now show that the prices on all markets increase after the positive price shock.
\begin{lemma}\label{lem:relationship1}
Let $x$ and $y$ be equilibria of a game $G$ before and after a price shock $(\delta_m)_{m\in M}$ with $\delta_m\geq 0$ for all $m\in M$. Then, 
on all markets $m\in M$ the price in $y$ is higher than in $x$, that is, $p_m^\delta(y_m)\geq p_m(x_m)$. 
\end{lemma}
\proof 
While in the proof of Lemma~\ref{lem:multiunique} we compared two equilibria of the same game, we now compare the equilibria before and after the price shock. The analysis remains largely unchanged: 
if we denote by $M^+:=\{m\in M:x_m+\frac{\delta_m}{r_m}<y_m\}$
the set of markets where the price decreases, then $M^+\neq\emptyset$ still implies that there is some firm $i$ with $\sum_{m\in M^+}{(y_{i,m}-x_{i,m})}>0$ and $y_i<x_i$ as in~\eqref{y_ileqx_i}, leading to the same contradiction as before. 
It follows that $M^+=\emptyset$, i.e. $p_m^\delta(y_m)\geq p_m(x_m)$ for all $m\in M$.
\endproof

Given increasing prices, all firms increase their quantity.
\begin{lemma}\label{lem:relationship2}
Let $x$ and $y$ be equilibria of a game $G$ before and after a price shock $(\delta_m)_{m\in M}$ with $\delta_m\geq 0$ for all $m\in M$. 
Then, each firm $i\in N$ produces more in $y$ than in $x$, $y_i\geq x_i$.
\end{lemma}
\proof 
Assume there is $i\in N$ with $y_i<x_i$. Then
\[ \pi^\delta_{i,m}(y)\leq c_i'(y_i)\leq c_i'(x_i)=\pi_{i,m}(x)\]
on every market $m\in M$ with $x_{i,m}>0$. With $p^\delta_m(y_m)\geq p_m(x_m)$ as found in the previous lemma, it follows that $y_{i,m}\geq x_{i,m}$ on every market with $x_{i,m}>0$, a contradiction to $y_i<x_i$. Hence $x_i\leq y_i$ for all $i\in N$.
\endproof

We now show that whenever a firm sells more quantity on every market in the new equilibrium,
then also the firm's profit increases.
\begin{lemma}\label{lem:minus}
Let $x$ and $y$ be equilibria of a game $G$ before and after a price shock $(\delta_m)_{m\in M}$ with $\delta_m\geq 0$ for all $m\in M$. If for a firm $i\in N$ it holds that $y_{i,m}\geq x_{i,m}$ for all $m\in M$, then $u_i(x)\leq u_i^\delta(y)$.
\end{lemma}
\proof
Recall that for every market $m$ with $x_{i,m}>0$ (or  $y_{i,m}>0$), we have
\begin{align*} p_m(x_m)-r_mx_{i,m}&=c'_i(x_i)\\
p_m^{\delta}(y_m)-r_my_{i,m}&=c'_i(y_i).
\end{align*}
Thus, we get
\begin{align}
\notag u_i^\delta(y)-u_i(x)&=\sum_{m\in M}p^\delta_m(y_m)y_{i,m}-p_m(x_m)x_{i,m}-c_i(y_i)+c_i(x_i)\\
 \notag &=c'_i(y_i)y_i -c'_i(x_i)x_i+\sum_{m\in M}r_m(y_{i,m}^2-x_{i,m}^2)-c_i(y_i)+c_i(x_i)\\
\label{ineq:ass} &\geq c'_i(y_i)y_i -c'_i(x_i)x_i-c_i(y_i)+c_i(x_i),
\end{align}
where we used in~\eqref{ineq:ass} that (by  assumption) $y_{i,m}\geq x_{i,m}$ for all $m\in M$.
Finally, 
\begin{align*}
c'_i(y_i)y_i -c'_i(x_i)x_i-c_i(y_i)+c_i(x_i)&=c'_i(y_i)(y_i-x_i)+c'_i(y_i)x_i -c'_i(x_i)x_i-c_i(y_i)+c_i(x_i)\\
&\geq c'_i(y_i)(y_i-x_i)-c_i(y_i)+c_i(x_i)\geq 0,
\end{align*}
where we use that $y_i\geq x_i$ and the assumption that $c_i$ is convex.
\endproof

We are now ready to prove the main theorem.
\proof{Proof of Theorem~\ref{thm:main}.}
Let $G$ be a game with $n$ firms. We show that for any given firm $i$, the inequality
$\frac{u_i^\delta(y)}{u_i(x)}\geq 1-\frac{(n-1)^2}{4n^2}$ holds. The theorem follows because $\gamma^u(G,\delta)$ is the minimum of this quantity across all firms.

 We denote by $M^-:=\{m\in M: y_{i,m}<x_{i,m}\}$ the set of markets where firm $i$ decreases its quantity after the price shock and similarly  by $M^+:=M\setminus M^-$ the set where $i$ increases its quantity after the price shock. Note that by Lemma~\ref{lem:minus} we can assume that $M^-\neq\emptyset$
 since otherwise we obtain already $u_i^\delta(y)\geq u_i(x)$.

 We assume that the markets are indexed such that market 1 is a solution to 
\begin{align}\label{eq:worstmarket}
 \argmin_{m\in M^-}{\frac{(p^\delta_m(y_m)-c_i'(x_i))x_{i,m}-r_m y_{i,m}(x_{i,m}-y_{i,m})}{(p_m(x_m)-c_i'(x_i))x_{i,m}}}.
\end{align}
Note that the denominator of the above fraction is always positive as any market $m\in M^-$ has a non-zero quantity $x_{i,m}>0$ and thus by the first order equilibrium condition also $p_m(x_m)>c_i'(x_i)$.

 We find the following relations that will be helpful later on: The quantity added in markets in $M^+$ corresponds exactly to the quantity taken away from markets in $M^-$ and the additional quantity $y_i-x_i$, i.e.
\begin{align}\label{eq:quant}
 \sum_{m\in M^+}{y_{i,m}-x_{i,m}}&=\sum_{m\in M^-}{x_{i,m}-y_{i,m}} + y_i-x_i.
\end{align}
Also, the price in markets $m\in M^+$ in $y$ is related to the marginal cost, i.e.
\begin{align}\label{eq:pricecost}
 p^\delta_m(y_m)&\geq p^\delta_m(y_m)-r_my_{i,m}=\pi_{i,m}^\delta(y)=c_i'(y_i)
\end{align}
 which in turn is related to the price on markets $m\in M^-$, i.e.
\begin{align}\label{eq:costprice}
 c_i'(y_i)&\geq \pi_{i,m}^\delta(y)=p^\delta_m(y_m)-r_my_{i,m}.
\end{align}
 As the cost function is convex,
\begin{align}\label{eq:costcon1}
 c_i'(y_i)(y_{i}-x_{i})\geq c_i(y_i)-c_i(x_i).
\end{align}

We combine the above to a statement that relates the profit of quantity added to $M^+$ to the profit lost by reducing quantity in $M^-$,
\begin{align}\label{eq:mpluscombined}
\sum_{m\in M^+}p^\delta_m(y_m)&(y_{i,m}-x_{i,m})\overset{\eqref{eq:pricecost}}{\geq} \sum_{m\in M^+}{c_i'(y_i)(y_{i,m}-x_{i,m})}\notag\\
 &\overset{\eqref{eq:quant}}{\geq} \sum_{m\in M^-}{c_i'(y_i)(x_{i,m}-y_{i,m})} + c_i'(y_i)(y_i-x_i)\notag\\ 
 &\overset{\eqref{eq:costprice},\eqref{eq:costcon1}}{\geq} \sum_{m\in M^-}{(p^\delta_m(y_m)-r_my_{i,m})(x_{i,m}-y_{i,m})} +c_i(y_i)-c_i(x_i).
\end{align}

We further assume that $\frac{u_i^\delta(y)}{u_i(x)}<1$, as we are interested in worst case instances and our lower bounds show that such instances exist. 
Note that for a fraction with value less than 1, subtracting the same amount from both numerator and denominator decreases the value of the fraction. We estimate 
\begin{align}
 \frac{u_i^\delta(y)}{u_i(x)}&=\frac{\sum_{m\in M^-}{p_m^\delta(y_m)y_{i,m}}+\sum_{m\in M^+}{p_m^\delta(y_m)y_{i,m}}-c_i(y_i)}{\sum_{m\in M^-}{p_m(x_m)x_{i,m}}+\sum_{m\in M^+}{p_m(x_m)x_{i,m}}-c_i(x_i)}\notag\\
  &\overset{\eqref{eq:mpluscombined}}{\geq}
    \frac{\sum_{m\in M^-}{p_m^\delta(y_m)x_{i,m}-r_my_{i,m}(x_{i,m}-y_{i,m})}+\sum_{m\in M^+}{p_m^\delta(y_m)x_{i,m}}-c_i(x_i)}{\sum_{m\in M^-}{p_m(x_m)x_{i,m}}+\sum_{m\in M^+}{p_m(x_m)x_{i,m}}-c_i(x_i)}\notag\\
  &\geq
    \frac{\sum_{m\in M^-}{(p_m^\delta(y_m)-c_i'(x_i))x_{i,m}-r_my_{i,m}(x_{i,m}-y_{i,m})}+\sum_{m\in M^+}{(p_m(x_m)-c_i'(x_i))x_{i,m}}}{\sum_{m\in M^-}{(p_m(x_m)-c_i'(x_i))x_{i,m}}+\sum_{m\in M^+}{(p_m(x_m)-c_i'(x_i))x_{i,m}}}\label{eq:estim1}\\
  &\geq
    \frac{\sum_{m\in M^-}{(p_m^\delta(y_m)-c_i'(x_i))x_{i,m}-r_my_{i,m}(x_{i,m}-y_{i,m})}}{\sum_{m\in M^-}{(p_m(x_m)-c_i'(x_i))x_{i,m}}}\label{eq:estim2}\\
  &\overset{\eqref{eq:worstmarket}}{\geq}
    \frac{(p_1^\delta(y_1)-c_i'(x_i))x_{i,1}-r_1y_{i,1}(x_{i,1}-y_{i,1})}{(p_1(x_1)-c_i'(x_i))x_{i,1}}\label{eq:majorsimplification}.
\end{align}
In~\eqref{eq:estim1} we use that the cost function is convex and hence $-c_i(x_i)\geq \sum_{m\in M}{-c_i'(x_i)x_{i,m}}$ and in \eqref{eq:estim2} that the price in a market with positive quantity is at least the marginal cost, i.e. $p_m(x_m)\geq c_i'(x_i)$ for a market with $x_{i,m}>0$.

We now further examine the relation of $p_1^\delta(y_1)$, $p_1(x_1)$ and $c_i'(x_i)$. For any firm $j\neq i$ that has increased its quantity on market $1$, i.e. $y_{j,1}>x_{j,1}$, we get
\[
  p_{1}(x_1)-r_{1}x_{j,1}=\pi_{j,1}(x)\leq c_j'(x_j)\leq c_j'(y_j)= \pi_{j,1}^\delta(y)= p_{1}^\delta(y_1)-r_{1}y_{j,1},  
\]
that is,
\begin{align}\label{eq:profitablej}
  r_1(y_{j,1}-x_{j,1})\leq  p_{1}^\delta(y_1)- p_1(x_1).
\end{align}
Then, considering $\sum_{j: y_{j,1}>x_{j,1}}{(y_{j,1}-x_{j,1})}+y_{i,1}-x_{i,1}\geq y_1-x_1$, we can rather precisely observe how the price on market $1$ changes with the price shock.
\begin{align*}
   p_{1}^\delta(y_1)&=p_1(x_1)+\delta_1-r_1(y_1-x_1)\geq p_1(x_1)+r_1(x_{i,1}-y_{i,1})-r_1\!\!\!\!\!\!\sum_{j: y_{j,1}>x_{j,1}}{\!\!\!\!\!\!(y_{j,1}-x_{j,1})}\notag\\
   &\overset{\eqref{eq:profitablej}}{\geq}p_1(x_1)+r_1(x_{i,1}-y_{i,1})-(n-1)(p_{1}^\delta(y_1)- p_1(x_1)),
\end{align*}
as there are at most $n-1$ firms with $y_{j,1}>x_{j,1}$. This can be rearranged to 
\begin{align}\label{eq:pricemarket2}
p_{1}^\delta(y_1)\geq p_1(x_1)+\frac{r_1}{n}(x_{i,1}-y_{i,1}).
\end{align}

Observe further that $x_{i,1}>0$ because market $1$ is in $M^-$ and hence 
\begin{align}\label{eq:minprice}
  p_1(x_1)-r_1x_{i,1}=\pi_{i,1}(x)= c_i'(x_i).
\end{align} 

We continue the proof from \eqref{eq:majorsimplification},
\begin{align}
\frac{u_i^\delta(y)}{u_i(x)}&\overset{\eqref{eq:majorsimplification}}{\geq}
    \frac{(p_1^\delta(y_1)-c_i'(x_i))x_{i,1}-r_1y_{i,1}(x_{i,1}-y_{i,1})}{(p_1(x_1)-c_i'(x_i))x_{i,1}}\notag\\
    &\overset{\eqref{eq:pricemarket2}}{\geq} \frac{(p_1(x_1)-c_i'(x_i))x_{i,1}+\frac{r_1}{n}(x_{i,1}-y_{i,1})x_{i,1}-r_1y_{i,1}(x_{i,1}-y_{i,1})}{(p_1(x_1)-c_i'(x_i))x_{i,1}}\notag\\
    &=1+\frac{r_1(\frac{1}{n}x_{i,1}-y_{i,1})(x_{i,1}-y_{i,1})}{(p_1(x_1)-c_i'(x_i))x_{i,1}}\notag\\
    &\overset{\eqref{eq:minprice}}{=}1+\frac{(\frac{1}{n}x_{i,1}-y_{i,1})(x_{i,1}-y_{i,1})}{x_{i,1}^2} \notag\\
    &\geq 1-\frac{(n-1)^2}{4n^2}= \frac{3}{4}+\frac{2n-1}{4n^2}.\notag
\end{align}
\endproof

\subsection{Lower Bound}
To show that the bound of the previous theorem is tight, we construct a simple instance with matching profit loss. 

\begin{proposition}\label{lem:lb_profit}
 For any $n$, there is a game $G$ with $n$ firms where a positive price shock decreases the profit of some firm $a$ by a factor $\frac{(n-1)^2}{4n^2}$, that is, 
\[ \gamma^u(G,\delta)= 1-\frac{(n-1)^2}{4n^2}.
\] 
\end{proposition}
\proof
 The instance has two markets $M=\{1,2\}$ and there are $n$ firms. All firms serve market $2$ while market $1$ is only served by some firm $a\in N$.
 We fix the price on market $1$ before the price shock to $0$, i.e. $p_1=0$ and $r_1=0$. On market $2$ the price is $p_2(q_2)=2-q_2$, where $q_2$ is the total quantity sold in market 2. The cost of firm $a$ for any total quantity $q_a=q_{a,1}+q_{a,2}$ is $0$ if $q_a\leq \frac{2}{n+1}$; for any larger quantity the cost is prohibitively high. Note that although we treat the cost function
of firm $a$ as a (non-differentiable) step function, we can also use differentiable
 and convex functions approximating this step
 function within any precision leading to the same bound (up to the selected precision). For firms $i\neq a$, the cost is always $0$.

 The Cournot equilibrium $x$ of this game can be found through convex optimization. In the equilibrium, no quantity is sold on market $1$ and on market $2$, $x_{a,2}=x_{i,2}=\frac{2}{n+1}$. Firm $a$'s profit is $\ut_a(x)=(2-n\frac{2}{n+1})\frac{2}{n+1}=\frac{4}{(n+1)^2}$.

 A price shock that increases the price on market $1$ to $\frac{n-1}{n^2}$ leads firm $a$ to shift to market $1$. In the new equilibrium $y$, $y_{a,1}= \frac{n-1}{n(n+1)}$, $y_{a,2}=\frac{1}{n}$, and $y_{i,2}=\frac{2n-1}{n^2}$. The profit of firm $a$ is
\begin{align*}
 \ut_a(y)&=\left(2-\frac{1}{n}-(n-1)\frac{2n-1}{n^2}\right)\frac{1}{n}+ \frac{n-1}{n^2} \frac{n-1}{n(n+1)} =\frac{3n-1}{n^2(n+1)}.\\
\end{align*}
 Then, the ratio of profit before and after the price shock is
\begin{align*}
  \gamma^u(G,\delta) =\frac{\ut_a(y)}{\ut_a(x)}&=\frac{(3n-1)(n+1)}{4n^2}=1-\frac{(n-1)^2}{4n^2}.
\end{align*}
\endproof

\begin{remark}
Note that this lower bound is quite generic in the sense that such an instance can be constructed for any price function on market 2 and any linear cost function for competitors $i\neq a$. In general, the profit loss of a firm can be large when it has a strongly convex cost function, such that a positive price shock in one market causes it to decrease quantity in another market, and when this is met by competitors with linear (or not ``too convex'') cost functions.
\end{remark}

\subsection{Non-convex Cost}
If we abandon Assumption~\ref{ass2}, i.e. allow non-convex cost functions, we possibly lose uniqueness of equilibria and we would have to redefine our objective function, e.g. involving equilibrium selection.\footnote{Cournot equilibria continue to exist for non-convex costs if inverse demand functions satisfy
rather mild assumptions, see, e.g., 
\citet{Novshek85} and 
\citet{Amir96} and 
\citet{Roberts76}.}
Moreover, one
can easily construct examples where a positive price shock completely eliminates the profit of firm a firm in all equilibria. If e.g. fixed costs are allowed, in the example of 
\citet{Bulow85} one can set the fixed cost of the monopolist equal to their revenue after the price shock using the fact that fixed costs do not change the equilibria of the game as long as nonnegative profits are guaranteed. Similar examples are possible if we mix between economies of scale and diseconomies of scale among firms' cost technologies.

\subsection{Concave Inverse Demand Functions}
Relaxing Assumption~\ref{ass1} toward concave prices reveals another counterintuitive phenomenon: Very small price shocks may decrease the profit of a firm by an arbitrary amount. 
Consider the class $\HA\supseteq\G$ of games that allows for concave inverse demand functions.
We obtain the following value for $\gamma^u(\HA,\Delta_{\HA})$.

\begin{proposition}
For any $k\geq 4$, there is a game with only two markets and concave price functions such that the profit ratio of one of the firms before and after a positive price shock is less than $\frac{2}{k}$.
Thus, $\gamma^u(\HA,\Delta_{\HA})=0$.
\end{proposition}
\proof
 Fix some $k\geq 4$. We construct a game $G\in\HA$ to fulfill the above claim.
 The firm whose profit ratio we observe, denoted by $a$, has cost $c_a(q_a)=0$ for quantities $q_a\leq 1$ and prohibitively high cost for larger quantities. All other firms $i\neq a$ have cost $c_i(q_i)= q_i$ for any quantity.
 
 Market $1$ is only served by firm $a$ and has constant price $p_1(q_1)\equiv 0$. Market $2$ is served by all firms and has a concave price function satisfying $p_2(1)=1$ with $p_2'(1)=-1$ and $p_2(1+\frac{1}{k})=1-\frac{2}{k}$ with $p_2'(1+\frac{1}{k})=-k$.

 The initial equilibrium $x$ of this game is $x_{a,1}=0$, $x_{a,2}=1$ and $x_{j,2}=0$ for all $j\neq a$. To verify this, observe that marginal revenue and cost of firm $a$ are all equal as $\pi_{a,1}(x)=0$, $\pi_{a,2}(x)=1-1=0$ and $c_a'(x)=0$ as well as for competitors $j\neq a$ we have $\pi_{j,2}(x)=1=c_j'(x)$.
 
 Let the price shock be $\delta_1=\frac{1}{k}$ and $\delta_2=\frac{3}{k}$. The new equilibrium is $y_{a,1}=1-\frac{1}{k}$, $y_{a,2}=\frac{1}{k}$ and $y_{j,2}=\frac{1}{k^2}$ for all competitors $j\neq a$. We again verify $\pi^\delta_{a,1}(y)=\frac{1}{k}$, $\pi^\delta_{a,2}(y)=1+\frac{1}{k}-k\frac{1}{k}=\frac{1}{k}$ which are equal and greater than $0$. For competitors $j\neq a$, $\pi^\delta_{j,2}(y)=1+\frac{1}{k}-k\frac{1}{k^2}=1=c_j'(y)$.
 
 We calculate the profit of firm $a$ in both equilibria: $u_a(x)=1$ and thus
 \begin{align*}
  \gamma(G,\delta)&=u_a^\delta(y) =p^\delta_1(y)y_{a,1}+p_2^\delta(y)y_{a,2}-c_a(y)\\
  &=\frac{1}{k}(1-\frac{1}{k})+(1+\frac{1}{k})\frac{1}{k}=\frac{2}{k}. 
 \end{align*}
\endproof

\section{Effect of Price Shocks on Aggregates}\label{sec:aggregate}
Theorem~\ref{thm:main} shows that an individual firm can lose no more than $25\%$ of its profit as a result of a positive price shock. The lower bound, however, had the property
that one firm loses but all competitors gained in their total profits.
In this section, we study effects of price shocks on \emph{aggregate} measures: the
total profit
and the social welfare.

\subsection{Total Profit}
Our first result showed that each firm may not lose more than  $25\%$
of current profits. By the definition of $\gamma^u(G,\delta)$ it follows that the same holds true for total firm's profits, that is, $\gamma^U(G,\delta)\geq \gamma^u(G,\delta)\geq 3/4$
for any game $G$.
More interestingly, we show an instance, where this loss is actually attained,
and, thus, $\gamma^U(G,\delta)=\gamma^u(G,\delta)$.
\begin{proposition}\label{lem:welfare}
There is a game $G\in\G$ with $n$ firms where a positive price shock $\delta$ decreases the equilibrium total profit by a $\frac{(n-1)^2}{4(n^2+n-1)}$ fraction of the original total profit, that is, by almost $25\%$ for instances with many firms, that is
\[\gamma^U(G,\delta) \leq 1-\frac{(n-1)^2}{4(n^2+n-1)}\rightarrow_{n\rightarrow\infty} \frac{3}{4}.
\]
\end{proposition}
\proof
The game is similar to that from that proof of Proposition~\ref{lem:lb_profit}, except that firm $a$ can produce a quantity of $q_a\leq 1$ at cost $0$ and its competitors $i\neq a$ have a per-unit cost of 1, i.e. the cost is $c_i(q_i)=q_i$.

 The Cournot equilibrium $x$ of this game can be found through convex optimization. In the equilibrium, no quantity is sold on market $1$ and on market $2$, $x_{a,2}=1$ and $x_{i,2}=0$ for all $i\neq a$. The equilibrium total profit is $U(x)=1$.

 A price shock that increases the price on market $1$ to $\delta=\frac{n^2-1}{2(n^2+n-1)}$ leads firm $a$ to shift to market $1$. In the new equilibrium $y$, $y_{a,1}= \frac{n^2-n}{2(n^2+n-1)}$, $y_{a,2}=\frac{n^2+3n-2}{2(n^2+n-1)}$, and $y_{i,2}=\frac{n-1}{2(n^2+n-1)}$ for all $i\neq a$. The new equilibrium total profit is
\begin{align*}
 U^\delta(y)&=\delta y_{a,1} + (2-(y_{a,2}+\sum_{i\neq a}y_{i,2}))(y_{a,2}+\sum_{i\neq a}y_{i,2})-\sum_{i\neq a}y_{i,2}\\
  &=\frac{3n^2+6n-5}{4(n^2+n-1)}.
\end{align*}
\endproof

\subsection{Aggregate Social Welfare}

We now consider the effect of price shocks on
social welfare defined as
\begin{align}\label{eq:surplusdef}
S(q)= \sum_{m\in M}\int_{0}^{q_m}p_m(z)dz-\sum_{i\in N}c_i(q_i)=\sum_{m\in M}\left(s_m q_m-\frac{r_mq_{m}^{2}}{2}\right)-\sum_{i\in N}c_i(q_i).
\end{align}
The first term in the above definition of $S(q)$ measures the value that the buyers in the market have of the goods while the second term simply sums up total production cost. Social welfare
has been considered before, among others, by 
\citet{AndersonR03}, 
\citet{Ushio85} and 
\citet{TsitsiklisX13}.
For a given game $G$, we want to bound the ratio
$\gamma^S(G,\delta)=\frac{S^\delta(y)}{S(x)}$.

\begin{theorem}\label{thm:surplus}
Given a game $G$, a positive price shock $\delta$ can decrease the social welfare by at most a factor $\frac{1}{6}$, that is,
\[ \gamma^S(G,\delta)\geq \frac{5}{6}.
\]
\end{theorem}
Before we prove the theorem, 
we characterize $y$ with the following \emph{variational inequality}.
Variational inequalities have been used before, e.g., in the context of 
characterizing equilibria (cf.~\citet{dafermos1980traffic,Haurie85})
and for
bounding the price of anarchy
(cf.~\citet{CorSS04,CCS06,CorreaSM08,Harks_Miller2011,Harks:stack2011,Rough2002,RS11a}).
\begin{lemma}\label{lem:variational}
Let $y$ be the equilibrium for the game with price shock $\delta_m\geq 0, m\in M$
and let $x$ be the original equilibrium.
Then, for all $i\in N$  it holds
\begin{align}\label{ineq:var} \sum_{m\in M_i}\big(s_m+\delta_m-r_m y_m-r_my_{i,m}-c'_i(y_i)\big)\big(x_{i,m}-y_{i,m}\big)\leq 0.\end{align}
\end{lemma}
\proof
For every firm $i$, given the equilibrium quantities $y_{-i}$ of its competitors, the problem
\[ \max_{q_{i,m}\geq 0, m\in M_i}u_i(q_i, y_{-i})\]
 is a convex program.
Thus, at an optimal solution $(y_{i,m})_{m\in M_i}$, the gradient $\nabla u_i(y)$
only decreases along any feasible direction. In particular, $\big(x_{i,m}-y_{i,m}\big)_{m\in M_i}$
is a feasible direction.
\endproof

We now bound the welfare gained at $x$.
\begin{lemma}\label{lem:surplus}
\begin{align}\label{eq:variational}
S(x)\geq 
 \sum_{i\in N}\sum_{m\in M}\frac{3}{2}r_m x_{i,m}^{2}.
\end{align}
\end{lemma}
\proof
Note that $x$ is an equilibrium for the unperturbed game.
In particular, the first order necessary optimality conditions
for each firm hold:
\begin{align}\label{eq:firstorder}
 s_m-r_m x_m-r_mx_{i,m}=c'_i(x_i) \text{ for all }x_{i,m}>0.
\end{align}
We combine this with the definition of the welfare~\eqref{eq:surplusdef},
\begin{align}
S(x)&=\sum_{m\in M}{\left(s_m x_m-\frac{r_mx_{m}^{2}}{2}\right)}-\sum_{i\in N}c_i(x_i) \notag\\
&\geq \sum_{m\in M}{\left(s_m x_m-\frac{r_mx_{m}^{2}}{2}\right)}-\sum_{i\in N}c'_i(x_i)x_i \label{ineq:convex}\\
&\overset{\eqref{eq:firstorder}}{=} \sum_{m\in M}{\left(s_m x_m-\frac{r_mx_{m}^{2}}{2}- (s_mx_m-r_m (x_m)^{2}-r_m \sum_{i\in N} (x_{i,m})^{2})\right)}\notag\\
&=\sum_{m\in M} \frac{r_mx_{m}^{2}}{2}+\sum_{i\in N}\sum_{m\in M}r_m x_{i,m}^{2}\geq \sum_{i\in N}\sum_{m\in M}\frac{3}{2}r_m x_{i,m}^{2}.\notag
\end{align}
Here \eqref{ineq:convex} follows by the convexity of $c_i$ and $c_i(0)=0$.
Finally, $x_m^2=\big(\sum_{i\in N}{x_{i,m}}\big)^2\geq \sum_{i\in N}{x_{i,m}^2}$.
\endproof

We now prove the theorem.
\proof{Proof of Theorem~\ref{thm:surplus}.}
We establish the difference between $S^\delta(y)$ and $S(x)$ through \eqref{eq:surplusdef} and the fact that $c_i(y_i)-c_i(x_i)\leq c_i'(y_i)(y_i-x_i)$ as the cost functions are convex and $y_i\geq x_i$ as shown in Lemma~\ref{lem:relationship2}.
\begin{align*}
S(x)&= S^\delta(y) + \sum_{m\in M}{\big(s_mx_m-(s_m+\delta_m)y_m -\frac{r_m}{2}(x_m^2-y_m^2)\big)}+\sum_{i\in N}{\big(c_i(y_i)-c_i(x_i)\big)}\\
&\leq S^\delta(y) + \sum_{m\in M}{\big((s_m+\delta_m)(x_m-y_m)-\frac{r_m}{2}(x_m^2-y_m^2)\big)}+\sum_{i\in N}{\big(c_i'(y_i)(y_i-x_i)\big)}.
\end{align*}
Subtracting the variational inequality~\eqref{ineq:var} summed up across all firms $i\in N$ allows to simplify the term.
\begin{align}
S(x)&\leq S^\delta(y) + \sum_{m\in M}{\big(-\frac{r_m}{2}(x_m^2+y_m^2)+r_my_mx_m+r_m \sum_{i\in N}{y_{i,m}(x_{i,m}-y_{i,m})}\big)}\notag\\
&=S^\delta(y) + \sum_{m\in M}{r_m\big(-\frac{1}{2}(x_m-y_m)^2+\sum_{i\in N}{y_{i,m}(x_{i,m}-y_{i,m})}\big)}\notag\\
&\leq S^\delta(y) + \sum_{m\in M}{r_m\sum_{i\in N}{\frac{x_{i,m}^2}{4}}}\leq S^\delta(y) + \frac{1}{6} S(x). \label{eq:binom}
\end{align}
Here, we used for the first inequality in~\eqref{eq:binom} that $(x_m-y_m)^2\geq 0$ and that $y_{i,m}(x_{i,m}-y_{i,m})\leq \frac{x_{i,m}^{2}}{4}$. The final step applies the result of Lemma~\ref{lem:surplus}.
\endproof

We can use the construction of Proposition \ref{lem:welfare}
to obtain a matching lower bound.

\begin{proposition}
 There is a game $G$ with many firms where a positive price shock decreases the 
social welfare by $1/6\approx16.6\%$, that is, 
\[\gamma^S(G,\delta)\ \overset{{n\rightarrow\infty}}{\rightarrow}\ \frac{5}{6}.\]
\end{proposition}
\proof
Note that $S(q)=U(q)+\frac{1}{2}\sum_{m\in M}{r_m q_m^2}$. We use the instance and equilibria from Proposition \ref{lem:welfare} and find $S(x) = U(x) + \frac{1}{2}1^2 = \frac{3}{2}$ and
\begin{align}
  S(y) &= U(y)+ \frac{1}{2}(y_{a,2}+\sum_{i\neq a}y_{i,2})^2 = \frac{10 n^4+22 n^3-7 n^2-24 n+11}{8 (n^2+n-1)^2}\ \overset{{n\rightarrow\infty}}{\rightarrow}\ 5/4.\notag 
\end{align}
Combined, this approaches $\gamma^S(G,\delta)\ \overset{{n\rightarrow\infty}}{\rightarrow}\ \frac{5}{4}\frac{2}{3}=\frac{5}{6}$.

%
%
\endproof

\section{Extension to Coarse Correlated Equilibria }\label{sec:unique}
In this section, we extend our results to a broader
set of equilibrium concepts including mixed, correlated and coarse correlated
equilibria.  It is well known that mixed equilibria are a subset of correlated equilibria
which itself are a subset of coarse correlated equilibria. We derive a structural result showing that
for our model
coarse correlated equilibria are essentially unique.
While for the class of strictly concave potential games (which includes our class of games),
uniqueness of correlated equilibria is already implied 
by the result of Neyman~\cite{Neyman97}\footnote{The uniqueness result of Neyman
for correlated equilibria has later been generalized to general concave games satisfying Rosen's~\cite{Rosen65} sufficient condition
for uniqueness, see Ui~\cite{Ui08}.},
our uniqueness result extends to coarse correlated equilibria. It is worth noting that 
there exist strictly concave (even quadratic) potential games that
allow for multiple coarse correlated equilibria, see
the recent results by Moulin et al.~\cite{Moulin14,MoulinRG14}.
We first define coarse correlated equilibrium (cf.~\cite{MoulinV78,Roughgarden09}).
\begin{definition}[Coarse correlated equilibrium]
 Given a strategic game $(N,X,u)$, a probability distribution $\sigma :X\rightarrow [0,1]$ is a \emph{coarse correlated equilibrium} if for every player $i\in N$,
\[
 \E{x\sim\sigma}{u_i(x)} \geq \E{x\sim\sigma}{u_i(y_i,x_{-i})}
\]
for all pure strategies $y_i\in X_i$.
\end{definition}

\begin{theorem}\label{thm:unique}
 For a (multimarket oligopoly) game $G$, let $x$ be a random variable drawn from the game's strategy space and let $\bx=\Exp{x}$ be its expected value. Then, the distribution of $x$ is a coarse correlated equilibrium if and only if $\bx$ is a pure Nash equilibrium, $x_m\overset{\text{a.s.}}{=}\bx_m$ on every market $m\in M$ and $c_i(\bx_i)=\Exp{c_i(x_i)}$ for every firm $i\in N$. In that case also $u_i(\bx)=\Exp{u_i(x)}$, $U(\bx)=\Exp{U(x)}$, and $S(\bx)=\Exp{S(x)}$.
\end{theorem}
We prove the theorem in four lemmata. We always denote by $x$ a random variable drawn from $X$ with $\bx=\Exp{x}$.
\begin{lemma}\label{lem:priceEq}
 If the distribution of $x$ is a coarse correlated equilibrium, then 
 $\bx_m\overset{\text{a.s.}}{=}x_m$ and, thus, $p_m(\bx_m)\overset{\text{a.s.}}{=}p_m(x_m)$ on every market $m\in M$.
\end{lemma}
\begin{proof}
 For a coarse correlated equilibrium by definition,
\begin{align}
  {\Exp{u_i(x)}}&\geq {\Exp{u_i(\bx_i,x_{-i})}}\notag\\
  &={\Exp{\sum_{m\in M_i}{p_m(\bx_{i,m}+x_{-i,m})\bx_{i,m}-c_i(\bx_i)}}}\notag\\
  &={\sum_{m\in M_i}{\Exp{p_m(\bx_{i,m}+x_{-i,m})}\bx_{i,m}-c_i(\bx_i)}}\notag\\
  &={\sum_{m\in M_i}{p_m(\bx_m)\bx_{i,m}-c_i(\bx_i)}}=u_i(\bx)\label{eq8}
\end{align}
where~\eqref{eq8} holds because the price functions are affine. Then, using that the cost is convex and hence $c_i(\bx_i)\leq \Exp{c_i(x_i)}$,
\begin{align}
 \sum_{m\in M}{\Exp{p_m(x_m)x_m}}&=\sum_{i\in N}{\Exp{u_i(x)+c_i(x_i)}}\notag\\
 &\geq \sum_{i\in N}{u_i(\bx) + c_i(\bx_i)}\notag\\
 &= \sum_{m\in M}{p_m(\bx_m)\bx_m}.\label{howe}
\end{align}
However, the markets' price functions are affine, hence, $p_m(x_m)x_m$ is concave yielding
\begin{equation}
\Exp{p_m(x_m)x_m}\leq p_m(\bx_m)\bx_m.\label{however}
\end{equation}
Combining~\eqref{howe} and~\eqref{however} implies that the quantity (and thus also the price) on every market is almost surely constant, that is, $x_m\overset{\text{a.s.}}{=}\bx_m$.
\end{proof}

\begin{lemma}\label{lem:costEq}
 If the distribution of $x$ is a coarse correlated equilibrium, then 
 $c_i(\bx_i)=\Exp{c_i(x_i)}$ for every firm $i\in N$.
\end{lemma}
\begin{proof}
\begin{align}
 \Exp{u_i(\bx_i,x_{-i})}&=\sum_{m\in M_i}{\Exp{p_m(\bx_{i,m}+x_{-i,m})}\bx_{i,m}}-c_i(\bx_i)\notag\\
  &=\sum_{m\in M_i}{\Exp{p_m(x_m)x_{i,m}}}-c_i(\bx_i)\label{isZero}\\
  &=\Exp{u_i(x)}+\Exp{c_i(x_i)}-c_i(\bx_i),\notag
\end{align}
where in~\eqref{isZero} we use that $p_m(x_m)$ is almost surely constant. As the distribution of $x$ is a coarse correlated equilibrium, i.e., $\Exp{u_i(\bx_i,x_{-i})}\leq\Exp{u_i(x)}$, and the cost functions are convex, i.e., $c_i(\bx_i)\leq \Exp{c_i(x_i)}$, we obtain $c_i(\bx_i)=\Exp{c_i(x_i)}$.
\end{proof}
Note that from Lemma~\ref{lem:priceEq} and Lemma~\ref{lem:costEq} it follows that if the distribution of $x$ is a coarse correlated equilibrium, $u_i(\bx)=\Exp{u_i(x)}$, $U(\bx)=\Exp{U(x)}$, and $S(\bx)=\Exp{S(x)}$.
\begin{lemma}
 If the distribution of $x$ is a coarse correlated equilibrium, then 
$\bx$ is a pure Nash equilibrium. 
\end{lemma}
\begin{proof}
Having shown that on every market $p_m(\bx_m)\overset{\text{a.s.}}{=}p_m(x_m)$ and for every firm $c_i(\bx_i)=\Exp{c_i(x_i)}$, we find that
\begin{align}
 u_i(\bx)&=\sum_{m\in M_i}{p_m(\bx_m)\bx_{i,m}}-c_i(\bx_i)\notag\\
  &=\sum_{m\in M_i}{\Exp{p_m(x_m)x_{i,m}}}-\Exp{c_i(x_i)}\notag\\
  &=\Exp{u_i(x)}.\notag
\end{align}
For every alternative strategy $y_i\in X_i$, the fact that the distribution of $x$ is a coarse correlated equilibrium gives 
\[
 \Exp{u_i(x)}\geq \Exp{u_i(y_i,x_{-i})},
\]
and because the price functions are affine we can use the linearity of expectation to find
\[
 \Exp{u_i(y_i,x_{-i})}= \sum_{m\in M_i}{\Exp{p_m(y_{i,m}+x_{-i,m})}y_{i,m}}-c_i(y_i) = u_i(y_i,\bx_{-i}).
\]
Combining the above, $u_i(\bx)\geq u_i(y_i,\bx_{-i})$.
\end{proof}

\begin{lemma}
 Let $x$ be distributed such that $\bx=\Exp{x}$ is a pure Nash equilibrium and $x_m\overset{\text{a.s.}}{=}\bx_m$ on every market $m\in M$ and $c_i(\bx_i)=\Exp{c_i(x_i)}$ for every firm $i\in N$. Then, the distribution of $x$ is a coarse correlated equilibrium.
\end{lemma}
\begin{proof}
  As above, $ \Exp{u_i(x)} = u_i(\bx)$ and $u_i(y_i,\bx_{-i}) = \Exp{u_i(y_i,x_{-i})}$ for all players $i$ and strategies $y_i\in X_i$. Also, $u_i(\bx) \leq u_i(y_i,\bx_{-i})$ because $\bx$ is a pure Nash equilibrium. Consequently, $\Exp{u_i(x)}\leq \Exp{u_i(y_i,x_{-i})}$, that is, the distribution of $x$ is a coarse correlated equilibrium.
\end{proof}
\begin{remark}
For the special case of one market and two firms (duopoly), the uniqueness result
stated in Theorem~\ref{thm:unique} is implied by a 
result of Gerard-Varet and Moulin~\cite{Varet78} where 
\emph{locally improvable strategies via coarse correlation} are characterized 
for a class of two-player games that includes the case of 
one market, two firms, affine prices and convex costs. For general
$n$-player games, however, not much is known regarding sufficient conditions for
uniqueness of coarse correlated equilibria~\cite{Moulin:private}. 
\end{remark}
For our robust quantitative comparative statics analysis, we obtain
that all our results carry over to mixed, correlated and coarse correlated equilibria.
\begin{corollary}
Given a game $G$ and a price shock $\delta_m, m\in M$ with $\delta_m\geq 0$
for all $m\in M$.  It holds that
\begin{enumerate}
\item $\gamma^u(G,\delta)\geq 3/4$
\item $\gamma^U(G,\delta)\geq 3/4$
\item $\gamma^S(G,\delta)\geq 5/6$,
\end{enumerate}
even if in the definition of $\gamma^t(G,\delta), t\in\{u,U,S\}$ mixed, correlated or coarse correlated equilibria
are considered.
\end{corollary}
\begin{remark}
While our efficiency bounds are \emph{robust} in the sense of Roughgarden~\cite{Roughgarden09}
(translating to a broader set of equilibrium concepts), our bound in Theorem~\ref{thm:main} is
not derived based on local smoothness arguments as used in~\cite{RS11a}. In fact,
in~\cite{RS11a} examples are shown (in the context of atomic splittable congestion games)
in which coarse correlated equilibria perform strictly worse than correlated equilibria.\end{remark}

\section{Conclusions}
\citet{Bulow85} showed that for multimarket oligopolies, a positive price shock can reduce a monopolist's profit. We directed a quantitative approach at their setting to provide an understanding of the significance and robustness of this effect.
Our results -- a positive price shock can reduce a profit by 25\% and a negative price shock can increase profit by 33\% -- imply that the effect may be significant. 
For example, the possible 33\% increase in profit may be enough for a government to consider imposing a domestic sales tax in order to force domestic companies to compete more aggressively abroad. For a market participant on the other hand, our results imply that equilibrium profit is robust in the sense that no more than 25\% of current profit is lost in case of a positive price shock. We further showed
that social welfare is more robust against positive price shocks: the worst case
loss is bounded by 16,6\%.

There are several natural extensions of our model:
\begin{itemize}[noitemsep,label=-]
\item How does the bounds change assuming non-linear price functions or bounded price shocks?
\item What can be said for more complicated market structures (nested in a graph)
as considered recently by Correa et al.~\cite{CorreaFLM14,CorreaLM14}?
\item What can be said for general aggregative games with strategic substitutes as considered by 
\citet{acemoglu2013}?
\item Bulow et al. show qualitatively a similar paradoxical effect for markets with strategic complements and joint economies of scale. Here, a quantification remains open.
\end{itemize}

%
%
%

\subsection*{Acknowledgments}
We thank two anonymous referees for their helpful comments
that have improved the presentation of the paper.
We further thank Henri de Belsunce, Amos Fiat, Paul Milgrom, Rudolf M\"uller, Hans Peters, {\'E}va Tardos
and the attendees of the Dagstuhl Seminar ``Auctions and Electronic Markets'' (November 2013)
as well as the attendees of the 9th Conference on Internet and Network
Economics 2013 (December 2013) 
for their insightful comments.
This research was supported by the Deutsche Forschungsgemeinschaft within the research training group 'Methods for Discrete Structures' (GRK 1408). The research of the first author was supported by the Marie-Curie grant ``Protocol Design (nr. 327546)'' funded within FP7-PEOPLE-2012-IEF.

\bibliographystyle{abbrvnat} 
\bibliography{../../master-bib} 

\begin{thebibliography}{70}
\providecommand{\natexlab}[1]{#1}
\providecommand{\url}[1]{\texttt{#1}}
\expandafter\ifx\csname urlstyle\endcsname\relax
  \providecommand{\doi}[1]{doi: #1}\else
  \providecommand{\doi}{doi: \begingroup \urlstyle{rm}\Url}\fi

\bibitem[Acemoglu and Jensen(2013)]{acemoglu2013}
D.~Acemoglu and M.~K. Jensen.
\newblock Aggregate comparative statics.
\newblock \emph{Games Econom. Behav.}, 81\penalty0 (0):\penalty0 27--49, 2013.

\bibitem[Amir(1996)]{Amir96}
R.~Amir.
\newblock {Cournot} oligopoly and the theory of supermodular games.
\newblock \emph{Games Econom. Behav.}, 15\penalty0 (2):\penalty0 132--148,
  1996.

\bibitem[Anderson and Renault(2003)]{AndersonR03}
S.~P. Anderson and R.~Renault.
\newblock Efficiency and surplus bounds in {Cournot} competition.
\newblock \emph{J. Econom. Theory}, 113\penalty0 (2):\penalty0 253--264, 2003.

\bibitem[Athey(2002)]{Athey02}
S.~Athey.
\newblock Monotone comparative statics under uncertainty.
\newblock \emph{The Quarterly Journal of Economics}, 117\penalty0 (1):\penalty0
  187--223, 2002.

\bibitem[Brander and Spencer(1985)]{brander1985}
J.~A. Brander and B.~J. Spencer.
\newblock Export subsidies and international market share rivalry.
\newblock \emph{Journal of International Economics}, 18\penalty0 (1):\penalty0
  83--100, 1985.

\bibitem[Bulow et~al.(1985)Bulow, Geanakoplos, and Klemperer]{Bulow85}
J.~Bulow, J.~Geanakoplos, and P.~Klemperer.
\newblock Multimarket oligopoly: Strategic substitutes and complements.
\newblock \emph{J. Polit. Econ.}, 93\penalty0 (3):\penalty0 488--511, 1985.

\bibitem[Cominetti et~al.(2009)Cominetti, Correa, and Stier-Moses]{CCS06}
R.~Cominetti, J.~R. Correa, and N.~E. Stier-Moses.
\newblock The impact of oligopolistic competition in networks.
\newblock \emph{Oper. Res.}, 57\penalty0 (6):\penalty0 1421--1437, 2009.

\bibitem[Corch\'on(1994)]{corchon1994}
L.~C. Corch\'on.
\newblock Comparative statics for aggregative games. the strong concavity case.
\newblock \emph{Math. Soc. Sci.}, 28\penalty0 (3):\penalty0 151--165, 1994.

\bibitem[Correa et~al.(2004)Correa, Schulz, and Stier-Moses]{CorSS04}
J.~R. Correa, A.~S. Schulz, and N.~E. Stier-Moses.
\newblock Selfish routing in capacitated networks.
\newblock \emph{Math. Oper. Res.}, 29\penalty0 (4):\penalty0 961--976, 2004.

\bibitem[Correa et~al.(2008)Correa, Schulz, and Moses]{CorreaSM08}
J.~R. Correa, A.~S. Schulz, and N.~E.~S. Moses.
\newblock A geometric approach to the price of anarchy in nonatomic congestion
  games.
\newblock \emph{Games Econom. Behav.}, 64\penalty0 (2):\penalty0 457--469,
  2008.

\bibitem[Correa et~al.(2014{\natexlab{a}})Correa, Figueroa, Lederman, and
  Moses]{CorreaFLM14}
J.~R. Correa, N.~Figueroa, R.~Lederman, and N.~E.~S. Moses.
\newblock Pricing with markups in industries with increasing marginal costs.
\newblock \emph{Math. Program.}, 146\penalty0 (1-2):\penalty0 143--184,
  2014{\natexlab{a}}.

\bibitem[Correa et~al.(2014{\natexlab{b}})Correa, Lederman, and
  Moses]{CorreaLM14}
J.~R. Correa, R.~Lederman, and N.~E.~S. Moses.
\newblock Sensitivity analysis of markup equilibria in complementary markets.
\newblock \emph{Oper. Res. Lett.}, 42\penalty0 (2):\penalty0 173--179,
  2014{\natexlab{b}}.

\bibitem[Dafermos(1980)]{dafermos1980traffic}
S.~Dafermos.
\newblock Traffic equilibrium and variational inequalities.
\newblock \emph{Transportation Sci.}, 14\penalty0 (1):\penalty0 42--54, 1980.

\bibitem[Debreu(1952)]{Debreu52}
G.~Debreu.
\newblock A social equilibrium existence theorem.
\newblock \emph{Proc. Natl. Acad. Sci. USA}, 38:\penalty0 886--893, 1952.

\bibitem[Dixit(1986)]{dixit1986}
A.~Dixit.
\newblock Comparative statics for oligopoly.
\newblock \emph{International Economic Review}, 27\penalty0 (1):\penalty0
  107--122, 1986.

\bibitem[Eaton and Grossman(1986)]{eaton1986}
J.~Eaton and G.~M. Grossman.
\newblock Optimal trade and industrial policy under oligopoly.
\newblock \emph{The Quarterly Journal of Economics}, 101\penalty0 (2):\penalty0
  383--406, 1986.

\bibitem[Edlin and Shannon(1998)]{Edlin98}
A.~S. Edlin and C.~Shannon.
\newblock Strict monotonicity in comparative statics.
\newblock \emph{J. Econom. Theory}, 81\penalty0 (1):\penalty0 201--219, 1998.

\bibitem[Fan(1952)]{Fan52}
K.~Fan.
\newblock Fixed point and minmax theorems in locally convex topological linear
  spaces.
\newblock \emph{Proc. Natl. Acad. Sci. USA}, 38:\penalty0 121--126, 1952.

\bibitem[Farahat and Perakis(2009)]{FarahatP09}
A.~Farahat and G.~Perakis.
\newblock Profit loss in differentiated oligopolies.
\newblock \emph{Oper. Res. Lett.}, 37\penalty0 (1):\penalty0 43--46, 2009.

\bibitem[Farahat and Perakis(2011)]{FarahatP11}
A.~Farahat and G.~Perakis.
\newblock Technical note - a comparison of {Bertrand} and {Cournot} profits in
  oligopolies with differentiated products.
\newblock \emph{Oper. Res.}, 59\penalty0 (2):\penalty0 507--513, 2011.

\bibitem[F{\'e}vrier and Linnemer(2004)]{fevrier2004}
P.~F{\'e}vrier and L.~Linnemer.
\newblock Idiosyncratic shocks in an asymmetric {Cournot} oligopoly.
\newblock \emph{International Journal of Industrial Organization}, 22\penalty0
  (6):\penalty0 835--848, 2004.

\bibitem[Gaudet and Salant(1991)]{gaudet1991}
G.~Gaudet and S.~W. Salant.
\newblock Increasing the profits of a subset of firms in oligopoly models with
  strategic substitutes.
\newblock \emph{The American Economic Review}, 81\penalty0 (3):\penalty0
  658--665, 1991.

\bibitem[Gerard-Varet and Moulin(1978)]{Varet78}
L.~Gerard-Varet and H.~Moulin.
\newblock Correlation and duopoly.
\newblock \emph{Journal of Economic Theory}, 19\penalty0 (1):\penalty0 123 --
  149, 1978.

\bibitem[Glicksberg(1952)]{Glicksberg52}
I.~Glicksberg.
\newblock A further generalization of the {Kakutani} fixed point theorem, with
  application to {Nash} equilibrium points.
\newblock \emph{Proc. Amer. Math. Soc.}, 3:\penalty0 170--174, 1952.

\bibitem[Guo and Yang(2005)]{Guo05}
X.~Guo and H.~Yang.
\newblock The price of anarchy of {Cournot} oligopoly.
\newblock In X.~Deng and Y.~Ye, editors, \emph{Proc. 1st Internat. Workshop on
  Internet and Network Econom.}, volume 3828 of \emph{LNCS}, pages 246--257,
  2005.

\bibitem[Harks(2011)]{Harks:stack2011}
T.~Harks.
\newblock Stackelberg strategies and collusion in network games with splittable
  flow.
\newblock \emph{Theory Comput. Syst.}, 48:\penalty0 781--802, 2011.

\bibitem[Harks and Miller(2011)]{Harks_Miller2011}
T.~Harks and K.~Miller.
\newblock The worst-case efficiency of cost sharing methods in resource
  allocation games.
\newblock \emph{Oper. Res.}, 59:\penalty0 1491--1503, 2011.

\bibitem[Haurie and Marcotte(1985)]{Haurie85}
A.~Haurie and P.~Marcotte.
\newblock On the relationship between {Nash-Cournot and Wardrop} equilibria.
\newblock \emph{Networks}, 15:\penalty0 295--308, 1985.

\bibitem[Johari and Tsitsiklis(2004)]{Johari04}
R.~Johari and J.~Tsitsiklis.
\newblock Efficiency loss in a network resource allocation game.
\newblock \emph{Math. Oper. Res.}, 29\penalty0 (3):\penalty0 407--435, 2004.

\bibitem[Johari and Tsitsiklis(2005)]{Johari05}
R.~Johari and J.~Tsitsiklis.
\newblock Efficiency loss in {Cournot} games.
\newblock Technical report, LIDS-P-2639, Laboratory for Information and
  Decision Systems, MIT, 2005.

\bibitem[Kakutani(1941)]{Kakutani41}
S.~Kakutani.
\newblock A generalization of {B}rouwer's fixed point theorem.
\newblock \emph{Duke Mathematics Journal}, 8\penalty0 (3):\penalty0 457--458,
  1941.

\bibitem[Kluberg and Perakis(2012)]{KlubergP12}
L.~J. Kluberg and G.~Perakis.
\newblock Generalized quantity competition for multiple products and loss of
  efficiency.
\newblock \emph{Oper. Res.}, 60\penalty0 (2):\penalty0 335--350, 2012.

\bibitem[Korilis et~al.(1999)Korilis, Lazar, and Orda]{Korilis99}
Y.~Korilis, A.~Lazar, and A.~Orda.
\newblock Avoiding the {Braess} paradox in noncooperative networks.
\newblock \emph{J. Appl. Probab.}, 36\penalty0 (1):\penalty0 211--222, 1999.

\bibitem[Koutsoupias and Papadimitriou(1999)]{Koutsoupias99}
E.~Koutsoupias and C.~Papadimitriou.
\newblock Worst-case equilibria.
\newblock In C.~Meinel and S.~Tison, editors, \emph{Proc. 16th Internat.
  Sympos. Theoretical Aspects of Comput. Sci.}, volume 1563 of \emph{LNCS},
  pages 404--413, 1999.

\bibitem[Krugman(1980)]{krugman1980}
P.~Krugman.
\newblock Scale economies, product differentiation, and the pattern of trade.
\newblock \emph{The American Economic Review}, 70\penalty0 (5):\penalty0
  950--959, 1980.

\bibitem[Kukushkin(1994)]{Kukushkin94}
N.~Kukushkin.
\newblock A fixed-point theorem for decreasing mappings.
\newblock \emph{Econ. Lett.}, 46:\penalty0 23--26, 1994.

\bibitem[Lin et~al.(2011)Lin, Roughgarden, Tardos, and Walkover]{LinRTW11}
H.~C. Lin, T.~Roughgarden, {\'E}.~Tardos, and A.~Walkover.
\newblock Stronger bounds on {Braess's} paradox and the maximum latency of
  selfish routing.
\newblock \emph{SIAM J. Discrete Math.}, 25\penalty0 (4):\penalty0 1667--1686,
  2011.

\bibitem[Melitz(2003)]{melitz2003}
M.~J. Melitz.
\newblock The impact of trade on intra-industry reallocations and aggregate
  industry productivity.
\newblock \emph{Econometrica}, 71\penalty0 (6):\penalty0 1695--1725, 2003.

\bibitem[Milgrom and Roberts(1990)]{Milgrom90}
P.~Milgrom and J.~Roberts.
\newblock Rationalizability, learning, and equilibrium in games with strategic
  complementarities.
\newblock \emph{Econometrica}, 58:\penalty0 1255--1277, 1990.

\bibitem[Milgrom and Roberts(1994)]{MR94}
P.~Milgrom and J.~Roberts.
\newblock Comparing equilibria.
\newblock \emph{American Economic Review}, 84\penalty0 (3):\penalty0 441--59,
  June 1994.

\bibitem[Milgrom and Shannon(1994)]{Milgrom94}
P.~Milgrom and C.~Shannon.
\newblock Monotone comparative statics.
\newblock \emph{Econometrica}, 62\penalty0 (1):\penalty0 157--80, 1994.

\bibitem[Monderer and Shapley(1996)]{Monderer96}
D.~Monderer and L.~Shapley.
\newblock Potential games.
\newblock \emph{Games Econom. Behav.}, 14\penalty0 (1):\penalty0 124--143,
  1996.

\bibitem[Moulin(2008)]{Moulin05}
H.~Moulin.
\newblock The price of anarchy of serial, average and incremental cost sharing.
\newblock \emph{Econom. Theory}, 36\penalty0 (3):\penalty0 379--405, 2008.

\bibitem[Moulin(2015)]{Moulin:private}
H.~Moulin.
\newblock private communication, 2015.

\bibitem[Moulin and Vial(1978)]{MoulinV78}
H.~Moulin and J.-P. Vial.
\newblock Strategically zero-sum games: The class of games whose completely
  mixed equilibria cannot be improved upon.
\newblock \emph{Internat. J. Game Theory}, 7\penalty0 (3-4):\penalty0 201--221,
  1978.

\bibitem[Moulin et~al.(2014{\natexlab{a}})Moulin, Ray, and Gupta]{Moulin14}
H.~Moulin, I.~Ray, and S.~S. Gupta.
\newblock Coarse correlated equilibria in an abatement game.
\newblock Cardiff Economics Working Papers E2014/24, Cardiff University, UK,
  2014{\natexlab{a}}.

\bibitem[Moulin et~al.(2014{\natexlab{b}})Moulin, Ray, and Gupta]{MoulinRG14}
H.~Moulin, I.~Ray, and S.~S. Gupta.
\newblock Improving {Nash} by coarse correlation.
\newblock \emph{J. Econom. Theory}, 150:\penalty0 852--865, 2014{\natexlab{b}}.

\bibitem[Neyman(1997)]{Neyman97}
A.~Neyman.
\newblock Correlated equilibrium and potential games.
\newblock \emph{Internat. J. Game Theory}, 26\penalty0 (2):\penalty0 223--227,
  1997.

\bibitem[Nisan et~al.(2007)Nisan, Roughgarden, Tardos, and
  Vazirani]{Nisan:2007}
N.~Nisan, T.~Roughgarden, {\'E}.~Tardos, and V.~Vazirani.
\newblock \emph{Algorithmic Game Theory}.
\newblock Cambridge University Press, Cambridge, UK, 2007.

\bibitem[Novshek(1985)]{Novshek85}
W.~Novshek.
\newblock On the existence of {Cournot} equilibrium.
\newblock \emph{Rev. Econ. Stud.}, 52\penalty0 (1):\penalty0 85--98, 1985.

\bibitem[Quah(2007)]{Quah07}
J.~K.-H. Quah.
\newblock The comparative statics of constrained optimization problems.
\newblock \emph{Econometrica}, 75\penalty0 (2):\penalty0 401--431, 03 2007.

\bibitem[Quirmbach(1988)]{quirmbach1988}
H.~C. Quirmbach.
\newblock Comparative statics for oligopoly: Demand shift effects.
\newblock \emph{International Economic Review}, 29\penalty0 (3):\penalty0
  451--459, 1988.

\bibitem[Roberts and Sonnenschein(1976)]{Roberts76}
J.~Roberts and H.~Sonnenschein.
\newblock On the existence of {Cournot} equilibrium without concave profit
  functions.
\newblock \emph{J. Econom. Theory}, 22:\penalty0 112--117, 1976.

\bibitem[Rosen(1965)]{Rosen65}
J.~Rosen.
\newblock Existence and uniqueness of equilibrium points in concave $n$-player
  games.
\newblock \emph{Econometrica}, 33\penalty0 (3):\penalty0 520--534, 1965.

\bibitem[Roughgarden(2002)]{Rough2002}
T.~Roughgarden.
\newblock The price of anarchy is independent of the network topology.
\newblock \emph{J. Comput. System Sci.}, 67:\penalty0 341--364, 2002.

\bibitem[Roughgarden(2006)]{Roug06}
T.~Roughgarden.
\newblock On the severity of {Braess}'s paradox: Designing networks for selfish
  users is hard.
\newblock \emph{J. Comput. System Sci.}, 72\penalty0 (5):\penalty0 922--953,
  2006.

\bibitem[Roughgarden(2009)]{Roughgarden09}
T.~Roughgarden.
\newblock Intrinsic robustness of the price of anarchy.
\newblock In M.~Mitzenmacher, editor, \emph{Proc. 41st Annual ACM Sympos.
  Theory Comput.}, pages 513--522, 2009.

\bibitem[Roughgarden and Schoppmann(2011)]{RS11a}
T.~Roughgarden and F.~Schoppmann.
\newblock Local smoothness and the price of anarchy in atomic splittable
  congestion games.
\newblock In \emph{Proc. 22nd Annual ACM-SIAM Sympos. on Discrete Algorithms},
  pages 255--267, 2011.

\bibitem[Roughgarden and Tardos(2002)]{Roughgarden02}
T.~Roughgarden and {\'E}.~Tardos.
\newblock How bad is selfish routing?
\newblock \emph{J. ACM}, 49\penalty0 (2):\penalty0 236--259, 2002.

\bibitem[Shannon(1995)]{Shannon95}
C.~Shannon.
\newblock Weak and strong monotone comparative statics.
\newblock \emph{J. Econom. Theory}, 5\penalty0 (2):\penalty0 209--27, 1995.

\bibitem[Topkis(1979)]{Topkis79}
D.~Topkis.
\newblock Equilibrium points in nonzero $n$-person submodular games.
\newblock \emph{SIAM J. Control Optim.}, 17:\penalty0 773--787, 1979.

\bibitem[Topkis(1998)]{Topkis98}
D.~Topkis.
\newblock \emph{Supermodularity and Complementarity}.
\newblock Princeton University Press, Princeton, NJ, USA, 1998.

\bibitem[Tsitsiklis and Xu(2012)]{Tsitsiklis14}
J.~N. Tsitsiklis and Y.~Xu.
\newblock Efficiency loss in a {Cournot} oligopoly with convex market demand.
\newblock Technical report, MIT, 2012.

\bibitem[Tsitsiklis and Xu(2013)]{TsitsiklisX13}
J.~N. Tsitsiklis and Y.~Xu.
\newblock Profit loss in {Cournot} oligopolies.
\newblock \emph{Oper. Res. Lett.}, 41\penalty0 (4):\penalty0 415--420, 2013.

\bibitem[Ui(2008)]{Ui08}
T.~Ui.
\newblock Correlated equilibrium and concave games.
\newblock \emph{Internat. J. Game Theory}, 37\penalty0 (1):\penalty0 1--13,
  2008.

\bibitem[Ushio(1985)]{Ushio85}
Y.~Ushio.
\newblock Approximate efficiency of {Cournot} equilibria in large markets.
\newblock \emph{Review of Economic Studies}, 52\penalty0 (4):\penalty0 547--56,
  1985.

\bibitem[Valiant and Roughgarden(2010)]{Valiant10}
G.~Valiant and T.~Roughgarden.
\newblock {Braess's} paradox in large random graphs.
\newblock \emph{Random Structures Algorithms}, 37\penalty0 (4):\penalty0
  495--515, 2010.

\bibitem[Vives(1990)]{Vives90}
X.~Vives.
\newblock Nash equilibrium with strategic complementarities.
\newblock \emph{J. Math. Econom.}, 19\penalty0 (3):\penalty0 305--321, 1990.

\bibitem[Vives(2005)]{Vives05}
X.~Vives.
\newblock Games with strategic complementarities: New applications to
  industrial organization.
\newblock \emph{Int. J. Ind. Organ.}, 23\penalty0 (7-8):\penalty0 625--637,
  2005.

\bibitem[Williamson and Shmoys(2011)]{williamson2011}
D.~P. Williamson and D.~B. Shmoys.
\newblock \emph{The design of approximation algorithms}.
\newblock Cambridge University Press, 2011.

\end{thebibliography}

\end{document}